\documentclass[journal]{IEEEtran}
\usepackage{amssymb}
\usepackage{amsmath}
\DeclareMathOperator*{\argmin}{argmin} % thin space, limits underneath in displays
\usepackage{latexsym}
\usepackage{epsfig}
\usepackage{verbatim}
\usepackage{mathrsfs}
\usepackage{float}
\usepackage{bigints}
\usepackage{graphics}
\usepackage{lscape}
\usepackage{enumitem}
\usepackage[symbol]{footmisc}

\newtheorem{theorem}{Theorem}[section]
\newtheorem{algorithm}{Algorithm}[section]
\newtheorem{lemma}{Lemma}[section]
\newtheorem{example}{Example}[section]
\newtheorem{definition}{Definition}[section]

\ifCLASSINFOpdf
\else
\fi
\begin{document}
%
% paper title
% Titles are generally capitalized except for words such as a, an, and, as,
% at, but, by, for, in, nor, of, on, or, the, to and up, which are usually
% not capitalized unless they are the first or last word of the title.
% Linebreaks \\ can be used within to get better formatting as desired.
% Do not put math or special symbols in the title.
\title{Effect of Pixelation on the Parameter Estimation of Single Molecule Trajectories}
%
%
% author names and IEEE memberships
% note positions of commas and nonbreaking spaces ( ~ ) LaTeX will not break
% a structure at a ~ so this keeps an author's name from being broken across
% two lines.
% use \thanks{} to gain access to the first footnote area
% a separate \thanks must be used for each paragraph as LaTeX2e's \thanks
% was not built to handle multiple paragraphs
%
\author{Milad R. Vahid,
        Bernard Hanzon,
        and~Raimund J. Ober,~\IEEEmembership{Senior Member,~IEEE}% <-this % stops a space
\thanks{This work was supported in part by the National Institutes of Health (R01 GM085575) and the Wellcome Trust (206411/Z/17/Z).}				
\thanks{Milad R. Vahid was with the Department of Biomedical Engineering, Texas A\&M University, College Station, TX, USA. He is now with the Department of Biomedical Data Science, Stanford University, Stanford, CA, USA e-mail: miladrv@stanford.edu.}% <-this % stops a space
\thanks{Bernard Hanzon is with the Department of Mathematics, University College Cork, Ireland e-mail: b.hanzon@ucc.ie.}% <-this % stops a space
\thanks{Raimund J. Ober is with the Centre for Cancer Immunology, Faculty of Medicine, University of Southampton, Southampton, UK e-mail: r.ober@soton.ac.uk.}}

\maketitle

% As a general rule, do not put math, special symbols or citations
% in the abstract or keywords.e

\begin{abstract}
The advent of single molecule microscopy has revolutionized biological investigations by providing a powerful tool for the study of intercellular and intracellular trafficking processes of protein molecules which was not available before through conventional microscopy. In practice, pixelated detectors are used to acquire the images of fluorescently labeled objects moving in cellular environments. Then, the acquired fluorescence microscopy images contain the numbers of the photons detected in each pixel, during an exposure time interval. Moreover, instead of having the exact locations of detection of the photons, we only know the pixel areas in which the photons impact the detector. These challenges make the analysis of single molecule trajectories, from pixelated images, a complex problem. Here, we investigate the effect of pixelation on the parameter estimation of single molecule trajectories. In particular, we develop a stochastic framework to calculate the maximum likelihood estimates of the parameters of a stochastic differential equation that describes the motion of the molecule in living cells. 
We also calculate the Fisher information matrix for this parameter estimation problem.
 The analytical results are complicated through the fact that the observation process in a microscope prohibits the use of standard Kalman filter type approaches. 
The analytical framework presented here is illustrated with examples of low photon count scenarios for which we rely on Monte Carlo methods to compute the associated probability distributions.
\end{abstract}

% Note that keywords are not normally used for peerreview papers.
\begin{IEEEkeywords}
Single molecule tracking, Pixelated detectors, Stochastic differential equations, Maximum likelihood estimation, Fisher information matrix, Cram\'er-Rao lower bound, Monte Carlo.
\end{IEEEkeywords}

% For peer review papers, you can put extra information on the cover
% page as needed:
% \ifCLASSOPTIONpeerreview
% \begin{center} \bfseries EDICS Category: 3-BBND \end{center}
% \fi
%
% For peerreview papers, this IEEEtran command inserts a page break and
% creates the second title. It will be ignored for other modes.
\IEEEpeerreviewmaketitle

\section{Introduction}

The study of intercellular and intracellular trafficking processes of objects of interest has been the subject of many research projects during the past few decades. The advent of single molecule microscopy made it possible to observe and track single molecules in living cells, which were not achievable before using conventional microscopes \cite{Moerner1,saxton1,Walter,Georg,ram2008,Kusumi,review2015}.

In fluorescence microscopy, the photons emitted by a fluorescently labeled object located in the object space are detected by a planar detector in the image space. In the fundamental data model, we assume that the time points and locations of the photons emitted by the object are detected by an ideal unpixelated detector. However, in practice, pixelated detectors, such as charge-coupled device (CCD) and electron multiplying CCD (EMCCD) cameras, are commonly used for acquiring the image of the object. In this case, referred to as the practical data model, the measurements, i.e., the fluorescence microscopy images, consist of the numbers of the photons detected in each pixel. Moreover, instead of having the exact locations of detection of the photons, we only know the pixel areas in which the photons impact the detector. These challenges make the analysis of single molecule trajectories from pixelated images a complex problem.

In the literature, there are several methods available concerning the problem of the parameter estimation of single molecule trajectories in cellular contexts. The majority of these methods model the effect of pixelation by using an additive noise in the fundamental data model. However, in general, this approximation does not describe the underlying stochastic model precisely. For example, in \cite{Michalet2010,Berglund2010,Berglund2012}, by encapsulating the effect of pixelation in a Gaussian additive random variable, referred to as the localization uncertainty, Berglund and Michalet have proposed methods for the estimation of diffusion coefficients based on mean square displacement of the observed locations of the molecule. For a similar observation model, Relich et al. \cite{Lidke} have proposed a method for the maximum likelihood estimation of the diffusion coefficient, with an information-based confidence interval, from Gaussian measurements. Although using these approximate observation models makes all corresponding computations simpler, it does not model the effect of the pixelated camera accurately. Calderon has extended Berglund's motion blur model to handle confined dynamics \cite{calderon1,calderon2,calderon3}. His proposed approach enables the estimation of the parameters of the motion model of the molecule by considering confinement and motion blur within a time domain maximum likelihood estimation framework. In \cite{andersson}, for the single molecule trajectory parameter estimation problem, a more accurate model has been used to describe the image of pixelated detectors. In this model, the expected intensity measured in each pixel is obtained by integrating the image profile, which is expressed in terms of a scaled and shifted version of the point spread function, over the pixel area. Here, we use a similar approach to model pixelated data more accurately.

In \cite{miladsiam}, we developed a stochastic framework in which we calculate the maximum likelihood estimates of the parameters of the model that describes the motion of the molecule in cellular environments. More importantly, we proposed a general framework to calculate the Cram\'er-Rao lower bound (CRLB), given by the inverse of the Fisher information matrix, for the estimation of unknown parameters and use it as a benchmark in the evaluation of the standard deviation of the estimates. In \cite{miladsiam}, we focused on the fundamental microscopy data model, in which the image of a molecule is acquired by an unpixelated detector. 

In this paper, we propose a general framework to investigate the effect of pixelation of the detector on the parameter estimation of single molecule trajectories accurately. We extend our previous results obtained for deterministic trajectories \cite{lin} and show examples of low photon count scenarios. We also include an example analysis in which we show how the standard deviation of parameter estimates depends on the pixel size of the detector. We consider complex relationships between the single molecule motion model, the photon emission process and the underlying statistical model of photon detection in the pixels of the detector in deriving all analytical expressions, e.g., likelihood function and Fisher information matrix, and no approximations are made.

This paper is organized as follows. In Section \ref{notations}, we introduce important notations used to define different data models in the paper. In Section \ref{datamodel}, we define fundamental and practical data models and present mathematical descriptions of them. In Section \ref{linear sde}, we introduce continuous-time stochastic differential equations, which are used to model the motion of single molecules in cellular environments, and calculate their solutions at discrete time points. Section \ref{mle} is devoted to calculation of the maximum likelihood estimates of the parameters of the system based on the introduced motion and data models in the previous sections. In Section \ref{crlbsection}, we calculate the general expressions for the Fisher information matrix for both of the fundamental and practical data model. In Section \ref{noise_effect_sec}, we investigate the effect of noise on the expressions derived in the previous sections. Finally, a summary of the paper and conclusions are provided in Section \ref{conclusions}.

\section{Notations}
\label{notations}
In this section, we introduce the following notations that will be used throughout the paper.

Let $\mathcal{C}^p$ be a pixelated detector defined as the union of a collection $\left\{C_1,\cdots,C_K\right\}$ of connected open and disjoint subsets of a region within $\mathbb{R}^2$ corresponding to the photon detection area of the detector. We use the notation $A_K^L=\{1,\cdots,K\}^L$ to denote the set of vectors of all possible pixel labels of $L$ photons detected in the pixelated image (see Section \ref{practical} for more detailed definitions). We also denote the Cartesian product of the pixel sets $C_{v_1},\ldots,C_{v_L}$, where $v_1, v_2,\cdots,v_L \in \{1,2,\ldots,K\}$ by $C_{v_{1:L}}=C_{v_1}\times C_{v_2} \times \ldots C_{v_L}$.
 
In this paper, for random vectors $X$ and $Y$, the conditional probability density function of $X$, given $Y$, is denoted by $p_{X|Y}$. For example, let  $X(\tau_1),\cdots,X(\tau_L)\in\mathbb{R}^3$ denote the locations of the molecule at a sequence of time points $(\tau_1,\cdots,\tau_L)\in\Delta^L:=\left\{(\tau_1,\ldots,\tau_L): t_0\leq \tau_1 <\tau_2 < \ldots <\tau_L \leq t\right\}$. Then, denoting the locations of the detected photons on the detector plane by $R_1:=U(X(\tau_1)),R_2:=U(X(\tau_2)),\cdots,R_L:=U(X(\tau_L))\in \mathbb{R}^2$, where $U$ is a random function that maps the object space into the image space, the conditional probability density function of $R_l$, given $R_{l-1}, l=2,\cdots,L$, is denoted by $p_{R_l|R_{l-1}}$. Note that in this paper, we use uppercase letters, e.g., $T_i$, to denote random variables, and lowercase letters, e.g., $\tau_i$, to denote particular values that the random variables can assume. We also denote the Poisson-distributed probability, with nonnegative intensity function $\Lambda(\tau), \tau\geq t_0$, of detecting $L$ photons in the time interval $[t_0,t]$ by $p_L:=\frac{e^{-\int_{t_0}^t\Lambda(\psi)d\psi}\left(\int_{t_0}^t\Lambda(\psi)d\psi\right)^L}{L!}$ (see Section \ref{fundamental} for more details).

Given the observed data $r\in\mathbb{R}^2$ with probability distribution $p^{\theta}_R$, where $\theta\in\mathbb{R}^n$ denotes the row vector of parameters, the Fisher information matrix $I(\theta)$ is given by
\begin{align}
I(\theta) = E\left\{\left( \frac{\partial}{\partial\theta}\log p^{\theta}_R(r)\right)^T\left( \frac{\partial}{\partial\theta}\log p^{\theta}_R(r)\right)\right\}.
\end{align} 
See Section \ref{crlbsection} for the detailed description of Fisher information matrix.

Also, in order to simplify vector representations of complex formulas in this paper, for arbitrary $n$-dimensional vectors $x_1,\cdots,x_L\in\mathbb{R}^n$ (or $n$-dimensional random variables $X_1,\cdots,X_L$), we define the $n\times L$ vector $x_{1:L}:=\left(x_1,\cdots,x_L\right)$ (or $X_{1:L}:=\left(X_1,\cdots,X_L\right)$). The integral of $f:\mathbb{R}^{n\times L}\mapsto\mathbb{R}$ over $\mathbb{R}^{n\times L}$ is denoted by $\int_{\mathbb{R}^{n\times L}}f(x_{1:L})dx_{1:L}$.

\section{Data model}
\label{datamodel}
In a standard optical microscope, the image of an object, which is in general moving in the object space, is captured by a detector in the image space.
In order to be able to model the photon detection process in a pixelated detector we firstly require a model of the precise impact locations on the detector of photons emitted by the imaged object.
Such a model, the \emph{fundamental data model}, has been developed in \cite{2004,2006,miladsiam}.
In this fundamental data model, we consider ideal conditions for the data acquisition procedure, in which it is assumed that we have an unpixelated image detector.
We briefly summarize this model here before proceeding to the development of the {\em practical data model}.
For this we use the fundamental data model to obtain a probabilistic description for the number of photons detected in each pixel of a pixelated detector.

\subsection{Fundamental data model}
\label{fundamental}

In the fundamental data model, the acquired data are the time points and locations of detection of the photons emitted from the object, where we have an unpixelated image detector. These time points and locations are intrinsically random. In general, the time points of detection of the emitted photons can be modeled as a counting process. The locations of detection of the photons emitted by the object are described by a random function that maps the object space into the image space. In this paper, for the fundamental data model, we assume that the locations of the photons, with the correct chronological ordering, emitted by the object can be detected by the detector.

We introduce the following notation. For $t_0\in\mathbb{R}$, let the random process $X(\tau), \tau\geq t_0$, describe the location of an object of interest, which emits photons, in the object space at time $\tau$. Let $\mathcal{C}:=\mathbb{R}^2$ denote a non-pixelated detector. Let $\left\{N(\tau), \tau\geq t_0\right\}$ be a Poisson process with non-negative and piecewise continuous intensity function $\Lambda(\tau), \tau\geq t_0$, that describes the time points of detection of the photons emitted by the object that impact the detector $\mathcal{C}$. These ordered time points, which are the events (jump points) of $\left\{N(\tau), \tau\geq t_0\right\}$, are denoted by one-dimensional (1D) random variables $t_0\leq T_1<T_2<\cdots$. The location of detection of the photon emitted by the object, at time $\tau\geq t_0$, that impacts the detector $\mathcal{C}$ is described by $U(X(\tau))$, where $U$ is a random function that maps the object space into the image space. For $x\in\mathbb{R}^3$, let $f_x$ denote the probability density function of $U(x)$, referred to as the \emph{image profile} of an object located at $x\in\mathbb{R}^3$ in the object space. In many practical scenarios, the image profile can be described as a scaled and shifted version of a function, referred to as the \emph{image function}, that describes the image of an object on the detector plane at unit lateral magnification. Assume that there exists a function $q_{z_0}{:}\ \mathbb{R}^2\mapsto\mathbb{R}, z_0\in\mathbb{R}$, such that for an invertible matrix $M\in\mathbb{R}^{2\times 2}$ and $x:=\left(x_0,y_0,z_0\right)\in\mathbb{R}^3$,
\begin{align}
\label{imagefunctionmain}
f_x\left(r\right):=\frac{1}{\left|\det\left(M\right)\right|}q_{z_0}\Big(M^{-1}r-(x_0,y_0)^T\Big),\quad r\in\mathcal{C}.
\end{align}

In particular, when the object is a point source and is in-focus with respect to the detector, according to optical diffraction theory, its image can be modeled by an Airy profile given by
\begin{align}
\label{airyprofileeq}
q(x_0,y_0)=\frac{J_1^2\left(\frac{2\pi n_a}{\lambda}\sqrt{x_0^2+y_0^2}\right)}{\pi\left(x_0^2+y_0^2\right)},\quad (x_0,y_0)\in\mathbb{R}^2,
\end{align}where $n_a$ denotes the numerical aperture of the objective lens, $\lambda$ denotes the emission wavelength of the molecule, and $J_1$ denotes the first order Bessel function of the first kind. In some applications, it is computationally more convenient to approximate the Airy profile by a Gaussian distribution given by
\begin{align}
\label{gaussianimagefunction}
q(x_0,y_0)=\frac{1}{2\pi\sigma^2}e^{-\frac{1}{2}\left(\frac{x_0^2+y_0^2}{\sigma^2}\right)},\quad (x_0,y_0)\in\mathbb{R}^2,
\end{align}where $\sigma>0$. 

For an out-of-focus point source, the image function can be modeled by the classical Born and Wolf model given by, for $(x_0,y_0)\in\mathbb{R}^2$, \cite{born}
\begin{small}
\begin{align}
\label{bornimagefunction}
q_{z_0}(x_0,y_0)=\frac{4\pi n_a^2}{\lambda^2}\left|\int_0^1J_0\left(\frac{2\pi n_a}{\lambda}\rho\sqrt{x_0^2+y_0^2}\right)e^{\frac{j\pi n_a^2z_0}{n_o\lambda}\rho^2}\rho d\rho\right|^2,
\end{align}
\end{small}where $J_0$ is the zeroth-order Bessel function of the first kind, $n_o$ is the refractive index of the objective lens immersion medium, and $z_0\in\mathbb{R}$ is the $z$-location of the point source on the optical axis in the object space.

The above mentioned image functions are only examples of the myriad of image functions that have been proposed and are used to describe image formation in a microscope and are themselves in fact approximations \cite{2004}. The subsequent developments are therefore carried out with a high level of generality to allow for the use of the most appropriate image function in a circumstance.

In \cite{miladsiam}, we use the fundamental data model for the maximum likelihood estimation of biophysical parameters such as diffusion and drift coefficients, from images acquired by an unpixelated detector. The joint probability density of the acquired data points, i.e. the impact locations of the detected photons on the detector, is the key probabilistic concept for our later developments. Specifically, in Theorem \ref{main0}, we calculate the conditional probability density function $p_{R_{1:L}|N(t)}$ of $R_1=U(X(T_1)),\cdots,R_L=U(X(T_L))$, given $N(t)$ for the fundamental data model. In the next section, we will use these results to characterize the acquired data from pixelated detectors.

\begin{theorem}
\label{main0}
The conditional probability density function $p_{R_{1:L}|N(t)}$ of $R_1,\cdots,R_L$, given $N(t)=L$, can be calculated as
\begin{small}
\begin{align}
\label{likelihood1}
&p_{R_{1:L}|N(t)}\left(r_{1:L}|L\right)\nonumber\\
&\ \ \ =\int_{\mathbb{R}^{3\times L}}\left(\int_{\Delta^L}F\left(r_{1:L},x_{1:L},\tau_{1:L}\right)d\tau_{1:L}\right) dx_{1:L},
\end{align}
\end{small}where
\begin{small} 
\begin{align}
\label{F_theta}
&F\left(r_{1:L},x_{1:L},\tau_{1:L}\right)\nonumber\\
&\ \ \ :=\frac{L!}{\left(\int_{t_0}^t\Lambda(\psi)d\psi\right)^L}\left(\prod_{i=1}^Lf_{x_i}(r_i)\Lambda(\tau_i)\right)p_{X(\tau_1),\cdots,X(\tau_L)}(x_{1:L}),
\end{align}
\end{small}is the conditional probability density function over the set of all $(r,x,\tau)$, given $N(t)=L$, and $p_{X(\tau_1),\cdots,X(\tau_L)}$ is the joint probability density function of $X(\tau_1),\cdots,X(\tau_L)$.

If $\left\{X(\tau_1),\cdots,X(\tau_L)\right\}$ is a Markov sequence, then,
\begin{small}
\begin{align*}
p_{X(\tau_1),\cdots,X(\tau_L)}\left(x_{1:L}\right)=p_{X(\tau_1)}\left(x_1\right)\prod_{l=1}^Lp_{X(\tau_l)|X(\tau_{l-1})}\left(x_l|x_{l-1}\right), 
\end{align*}
\end{small}where $p_{X(\tau_l)|X(\tau_{l-1})}, l=2,\cdots,L$, is the conditional probability density function of $X(\tau_l)$, given $X(\tau_{l-1})$, and $p_{X(\tau_1)}$ is the probability density function of $X(\tau_1)$.
\end{theorem}

\begin{proof}
See Appendix \ref{proof_main0}.$\hfill\Box$ 
\end{proof}

\subsection{Practical data model}
\label{practical}

In practice, pixelated detectors, e.g. CCD and EMCCD cameras, are commonly used for acquiring images of fluorescently labeled objects. In this subsection, we describe the practical data model.

\begin{figure}[htbp]
\centering\includegraphics[width=0.5\textwidth]{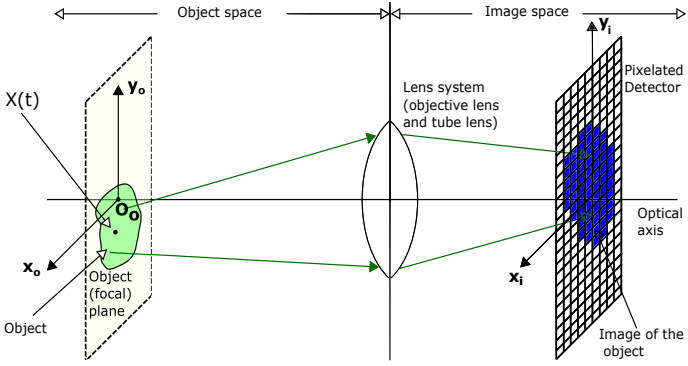}
\DeclareGraphicsExtensions{.eps} \caption{Schematic of an optical microscope.
An object located in the object (focal) plane is imaged by an optical lens system and the image of the object is acquired by a pixelated detector in the image space. A random variable $X(t), t\geq t_0$, describes the location of the object in the object plane at time $t$.}
\label{drawing}
\end{figure}

In the practical data model, the data acquired by a pixelated detector are the number of detected photons at each pixel (Fig. \ref{drawing}). Let the pixelated detector $\mathcal{C}^p$ be defined as the union of a collection $\left\{C_1,\cdots,C_K\right\}$ of connected open and disjoint subsets of a region within $\mathbb{R}^2$ corresponding to the photon detection area of the detector, such that $\bigcup_{k=1}^KC_k=\mathcal{C}^p$. We use the random variable $S_k, k=1,\cdots,K$, to describe the number of photons in the pixel $C_k$ that result from the detection of photons from the object of interest. Let $L=\sum_{k=1}^{K} S_k$ denote the total number of photons that impact the detector in a given experiment. We now need a notation that allows us to translate the information of the pixel labels of the individual photons to the number of photons in each pixel (i.e. the $S_k$ values). Note that $L=0$ if and only if $S_k=0$ for all $k=1,2,\cdots,K$, which is a trivial case. Now assume $L\geq 1$.

Let $A_K^L=\{1,\cdots,K\}^L$ denote the set of vectors of all possible pixel labels of $L$ photons. This means that each $L$ element vector in $A_K^L$ contains the pixel numbers in which each of the detected photons is captured. For the simple example where we have $K=2$ pixels and $L=3$ photons, for each vector in the set $A^3_2$, the first component denotes the pixel in which the first detected photon is captured, etc. The vector $v=(2,1,2)\in A^3_2$ implies that the first photon is captured by pixel $2$, the second photon by pixel $1$ and the third photon again by pixel $2$. The set of all such vectors for this example is given by
\begin{small} 
\begin{align*}
A^3_2:&=\left\{(1,1,1),(1,1,2),(1,2,1),(1,2,2),(2,1,1),\right.\nonumber\\
&\left.\ \ \ \ \ \ \ \ \ \ \ \ (2,1,2),(2,2,1),(2,2,2)\right\}.
\end{align*}
\end{small}
Note that for a vector $v=(v_1,v_2,\dots,v_L)$,
\begin{footnotesize}
\begin{align}
\label{new_example}
&Pr\left[\bigcap_{l=1}^L\left(R_l\in C_{v_l}\right)|N(t)=L\right]=\int_{C_{v_{1:L}}}p_{R_{1:L}|N(t)}\left(r_{1:L}|L\right)dr_{1:L},
\end{align}
\end{footnotesize}
or for our small example, for $v=(v_1,v_2,v_3)=(2,1,2)$,
\begin{small}
\begin{align}
%Pr\left[\bigcap_{l=1}^3\left(R_l\in C_{v_l}\right)|N(t)=3 \right]
&Pr\left[\left(R_1\in C_{2}\right) \cap \left(R_2\in C_{1}\right) \cap \left(R_3\in C_{2}\right) |N(t)=3 \right]\nonumber\\
&\ \ \ \ \ \ =\int_{C_{2}}\int_{C_{1}}\int_{C_2}p_{R_1,R_2,R_3|N(t)}\left(r_1,r_2,r_3|3\right)dr_3 dr_2 dr_1.
\end{align}
\end{small} 

In order to be able to compute the probability $Pr[S_1=z_1,S_2=z_k,\dots,S_K=z_k]$, we need to know all the events that lead to the photon count $S_1=z_1,S_2=z_k,\dots,S_K=z_k$, or we need to know the elements of $A_K^L$ that can lead to this photon count. For our example, as we have $K=2$ pixels we are interested to determine $Pr[S_1=z_1,S_2=z_k]$. If as above we know that $L=z_1+z_2=3$, i.e. that sum of the pixel counts is $3$, we need to determine the elements of $A_2^3$ that are consistent with this photon count.
 
We therefore, for a vector $v\in A^L_K$, denote by $\left\|v\right\|_{=k}, k=1,\cdots,K$, the number of the entries of $v$ which are equal to $k$. So as $v$ denotes a vector of pixel labels of $L$ photons, then  $\left\|v\right\|_{=k}, k=1,\cdots,K$ denotes the number of photons that have ended up in pixel $k$. For example, for $v=(1,1,2)\in A^3_2$, we have $\left\|v\right\|_{=1}=2, \left\|v\right\|_{=2}=1$, and $\left\|v\right\|_{=3}=0$.

This leads us to the following notation, which, for given photon numbers $z_1,\cdots,z_K\in\left\{0,1,\cdots\right\}$, and $\sum_{k=1}^Kz_k=L$, lists the elements in $A_K^L$ that will produce these photon counts in the pixels $C_1,C_2,\dots,C_K$, i.e.
 \begin{align}
A^L_K\left(z_1,\cdots,z_K\right):=\left\{v\in A^L_K|\left\|v\right\|_{=k}=z_k, k=1,\cdots,K\right\}.
\end{align}
So this set contains all vectors of pixel labels of $L$ photons that correspond to a configuration in which $z_k$ photons have landed in pixel $C_k,$ for $k=1,2,\ldots,K.$

Continuing our example we want to determine $A^3_2(1,2)$, i.e. the elements in $A^3_2$ that are such that there is one photon detected in pixel $C_1$ and $2$ photons in pixel $C_2$. Examining the set $A^3_2$ we see that $(1,2,2)$ is an element of $A^3_2(1,2)$ as one photon, the first detected, is captured in pixel $1$ and two photons, the second and third detected photon, are captured in pixel $2$. Proceeding in this manner we obtain, 
\begin{small} 
\begin{align*}
%&A^3_2:=\left\{(1,1,1),(1,1,2),(1,2,1),(1,2,2),(2,1,1),(2,1,2),(2,2,1),(2,2,2)\right\},\nonumber\\
&A^3_2\left(1,2\right):=\left\{(1,2,2),(2,1,2),(2,2,1)\right\},\nonumber\\
&A^3_2\left(2,1\right):=\left\{(1,1,2),(1,2,1),(2,1,1)\right\},\nonumber\\
&A^3_2\left(0,3\right):=\left\{(2,2,2)\right\},\nonumber\\
&A^3_2\left(3,0\right):=\left\{(1,1,1)\right\}. 
\end{align*}
\end{small}
Moreover, for example, $A^3_2 (4,0)$ is an empty set, as by assumption we have $3$ photons detected  and therefore there is no event in $A^3_2$ that leads to a count of $4$ in pixel $1$.

With this notation we can now immediately see what one needs to do to, for example, compute the probability $Pr[S_1=2,S_2=1]$. The set of events that lead to $S_1=2$ and $S_2=1$ is given by $A^3_2(2,1)$, i.e. $\left \{(1,1,2),(1,2,1),(2,1,1)\right\}$. Therefore,
\begin{footnotesize}
\[
Pr[S_1=2,S_2=1]  
= Pr\left [
\left(\left(R_1\in C_{1}\right) \cap \left(R_2\in C_{1}\right) \cap \left(R_3\in C_{2}\right)\right) 
\right .
\]
\[
\cup
\left(\left(R_1\in C_{1}\right) \cap \left(R_2\in C_{2}\right) \cap \left(R_3\in C_{1}\right)\right)
\]
\[ 
\left .
\cup
\left(\left(R_1\in C_{2}\right) \cap \left(R_2\in C_{1}\right) \cap \left(R_3\in C_{1}\right)\right)  | N(t)=L
\right] Pr \left[ N(t)=L \right]
\]
\[
= Pr
\left [
 \bigcup _{v_{1:3} \in A^3_2(2,1)} 
 \left(\bigcap_{l=1}^3R_l\in C_{v_l}\right) | N(t)=3 
\right] Pr \left[ N(t)=3 \right]
\]
\[
= \sum_{v_{1:3} \in A^3_2(2,1)} Pr \left [
\left(\bigcap_{l=1}^3R_l\in C_{v_l}\right) | N(t)=3
\right] Pr \left[ N(t)=3 \right].
\]
\end{footnotesize}
The expressions on the right hand side can be computed using the expression above (Eq. (\ref{new_example})).

In the following theorem we summarize the above derivations and state the main result of this section. It provides the desired expression for the discrete multivariable probability distribution $Pr[S_1=z_1,S_2=z_2,\dots,S_K=z_K]$ that $z_1,z_2,\dots,z_K$ photons are detected in pixels $C_1,C_2,\dots,C_K$. Assuming that we have acquired $L$ photons in total on the detector, this expression is given by Eq. (\ref{practical_eq2}) below. It shows how we can compute the probability that the detected photons during one trajectory impact the various pixels. In this theorem we need to make a technical distinction between two different versions of the detector. If $\bar{\mathcal{C}^p}=\mathbb{R}^2$, the model is referred to as the \emph{infinite practical data model}. We will see later that the infinite practical model will lead to somewhat simplified expressions, for example of the likelihood function. Obviously, in reality no practical detector will be infinite. But this model can serve as a useful approximation in the cases where we can assume that all photons that impact the infinite detector plane that is spanned by the detector, are in fact captured by the detector itself.
 
 \vskip 2mm
 
\begin{theorem}
\label{main}
Let $t_0$ and $t, t_0<t$, be given. 1. In the infinite practical data model, we have, for $z_1,\cdots,z_K=0,1,\cdots$,
\begin{footnotesize}
\begin{align}
\label{practical_eq2}
&Pr\left[S_1=z_1,\cdots,S_K=z_K\right]=p_L\sum_{v_{1:L}\in A^L_K\left(z_{1:K}\right)}\int_{C_{v_{1:L}}}\Bigg[\int_{\mathbb{R}^{3\times L}}\nonumber\\
&\Bigg(\int_{\Delta^L}F\left(r_{1:L},x_{1:L},\tau_{1:L}\right)d\tau_{1:L}\Bigg)dx_{1:L}\Bigg]dr_{1:L},
\end{align}
\end{footnotesize}where $L=\sum_{k=1}^Kz_{k}$, and $F(.)$ is given by Eq. (\ref{F_theta}).

2. In the practical data model, we have
\begin{footnotesize}
\begin{align}
\label{practical_eq}
&Pr\left[S_1=z_1,\cdots,S_K=z_K\right]=\sum_{z=0}^{\infty}p_{L+z}\sum_{v_{1:L+z}\in A^{L+z}_{K+1}\left(z_{1:K},z\right)}\nonumber\\
&\int_{C_{v_{1:L+z}}}\Bigg[\int_{\mathbb{R}^{3\times (L+z)}}\Bigg(\int_{\Delta^{L+z}}F\left(r_{1:L+z},x_{1:L+z},\tau_{1:L+z}\right)d\tau_{1:L+z}\Bigg)\nonumber\\
&\times dx_{1:L+z}\Bigg]dr_{1:L+z}.
\end{align}
\end{footnotesize}

Here $C_{K+1}$ denotes the complement (in the detector plane) of the closure of the union of the pixel sets that form the detector, so it accounts for all the photons that have gone through the detector plane, but have missed the detector (Note that $C_{K+1}$ will be open, but we do not claim $C_{K+1}$ is connected).
\end{theorem}

\begin{proof}
See Appendix \ref{proof_main}.$\hfill\Box$ 
\end{proof}

The result states that in order to compute the discrete probability distribution for a set of photon counts we need to carry out a sum of integrals. The number of summands is given by the size $\left|A^L_K\left(z_1,\cdots,z_K\right)\right|$ of the set $A^L_K\left(z_1,\cdots,z_K\right)$ which is equal to $\frac{L!}{z_1!\cdots z_K!}$.

In the above theorem, we account for all photons that cross the detection plane, and therefore the time points of detection are Poisson distributed with intensity function $\Lambda(\psi),~\psi \in [t_0,t].$ Another approach would be to only use the time points at which a photon is detected on the detector. These time points still form a Poisson process, but with a location-sensitive intensity function. This approach leads to a very complicated analysis, which is why we take the alternative approach as in Theorem \ref{main}.

\section{Linear stochastic systems}
\label{linear sde}

In general, the motion of an object in cellular environments is subject to different types of forces, e.g., deterministic forces due to the environment and random forces due to random collisions with other objects \cite{Schuss,kervrann}. The 3D random variable $X(\tau)$ denotes the location of the object at time $\tau\geq t_0$. Then, the motion of the object is assumed to be modeled through a general state space system with state $\tilde{X}(\tau)\in\mathbb{R}^k, \tau\geq t_0$, as, for $\tau_0:=t_0\leq\tau_1<\cdots<\tau_{l+1}<\cdots$,  
\begin{align}
\label{sde1}
\tilde{X}(\tau_{l+1})=\tilde{\phi}(\tau_l,\tau_{l+1})\tilde{X}(\tau_l)+\tilde{W}(\tau_l,\tau_{l+1}),
\end{align}where we assume that there exists a matrix $H\in\mathbb{R}^{3\times k}$ such that $X(\tau)=H\tilde{X}(\tau), \tau\geq t_0$, $\tilde{\phi}(\tau_l,\tau_{l+1})\in\mathbb{R}^{k\times k}$ is a state transition matrix, and $\left\{\tilde{W}(\tau_l,\tau_{l+1})\right.$, $\left. l=1,2,\cdots\right\}$ is a sequence of $k$-dimensional random variables. We also assume that the initial state $\tilde{X}(t_0)$ is independent of $\tilde{W}$ and its probability density function is given by $p_{\tilde{X}(t_0)}$. The framework we employ here is very general in that, depending on the specific problem we are considering, the system matrices are assumed to be known, unknown or partially known. The unknown elements of the system matrices would form part of the parameter vector that is to be estimated.

The general system of discrete evolution equations described by Eq. (\ref{sde1}) can arise, for example, from stochastic differential equations \cite{fokker}. In particular, in many biological applications, solutions of linear stochastic differential equations are good fits to experimental single-molecule trajectories \cite{fokker}. As an example, we assume that the motion of the object of interest, e.g., a single molecule, is described by the following linear vector stochastic differential equation \cite{calderon3}
\begin{align}
\label{lineareq}
dX(\tau)=\left(V+F(\tau)X(\tau)\right)d\tau+G(\tau)dB(\tau),\quad\quad\tau\geq t_0,
\end{align}where the 3D random process $X(\tau)$ describes the location of the object at time $\tau\geq t_0$, $F\in\mathbb{R}^{3\times 3}$ and $G\in\mathbb{R}^{3\times r}$ are continuous matrix time-functions related to the first order drift and diffusion coefficients, respectively, $V\in\mathbb{R}^3$ is the zero order drift coefficient, and $\left\{B(\tau)\in\mathbb{R}^r,\tau\geq t_0\right\}$ is an $r$-vector Brownian motion (Wiener) process with $E\left\{dB(\tau)dB(\tau)^T\right\}=I_{r\times r}, \tau\geq t_0$, where $I_{r\times r}$ is the $r\times r$ identity matrix \cite{calderon1,calderon2,calderon3}. Then, the solution of Eq. (\ref{lineareq}) at discrete time points $\tau_0:=t_0\leq\tau_1<\cdots<\tau_{l+1}<\cdots$ is given by \cite{jazwinski}
\begin{align}
\label{sdeConstant}
X(\tau_{l+1})=\phi(\tau_l,\tau_{l+1})X(\tau_l)+a(\tau_l,\tau_{l+1})+W(\tau_l,\tau_{l+1}),
\end{align}where the continuous matrix time-function $\phi\in\mathbb{R}^{3\times 3}$ is given by
\begin{align*}
&\frac{d\phi(t,\tau)}{dt}=F(t)\phi(t,\tau),\quad \phi(\tau,\tau)=I_{3\times 3},\quad \mbox{for all}\ t,\tau\geq t_0,\\
&\phi(t,\tau)\phi(\tau,\psi)=\phi(t,\psi),\quad \mbox{for all}\ t, \tau, \psi\geq t_0,
\end{align*}
and the vector $a(\tau_l,\tau_{l+1})\in\mathbb{R}^{3\times 1}$ is given by
\begin{align*} 
a(\tau_l,\tau_{l+1}):=V\int_{\tau_l}^{\tau_{l+1}}\phi(\tau,\tau_{l+1})d\tau.
\end{align*}
Also, in this case, 
\begin{small}
\begin{align*}
\left\{W(\tau_l,\tau_{l+1}):=\int_{\tau_l}^{\tau_{l+1}}\phi(\tau,\tau_{l+1})G(\tau)dB(\tau), l=1,2,\cdots\right\}
\end{align*}
\end{small}is a zero mean white Gaussian sequence with covariance $Q(\tau_l,\tau_{l+1})\in\mathbb{R}^{3\times 3}$ given by
\begin{align*}
Q(\tau_l,\tau_{l+1})=\int_{\tau_l}^{\tau_{l+1}}\phi(\tau,\tau_{l+1})G(\tau)G^T(\tau)\phi^T(\tau,\tau_{l+1})d\tau.
\end{align*}
By letting $X(\tau)=H\tilde{X}(\tau)=I_{3\times 3}\tilde{X}(\tau)=\tilde{X}(\tau),\tau\geq t_0$, and $\phi(\tau_l,\tau_{l+1})=\tilde{\phi}(\tau_l,\tau_{l+1}),$ we obtain expressions of the form of Eq. (\ref{sde1}), where we assume that
\begin{align*} 
\left\{\tilde{W}(\tau_l,\tau_{l+1})=a(\tau_l,\tau_{l+1})+W(\tau_l,\tau_{l+1}), l=1,2,\cdots\right\}
\end{align*}
is a white Gaussian sequence with mean $a(\tau_l,\tau_{l+1})$ and covariance $Q(\tau_l,\tau_{l+1})$.

\section{Maximum likelihood estimation}
\label{mle}

In this section, we provide a general framework to calculate the maximum likelihood estimates of the parameters of interest for both fundamental and practical data models. In general, these parameters can include the ones that describe the motion of the object, such as drift and diffusion coefficients, or the ones related to the image formation of the object on the detector, such as the intensity function. In the following, we briefly explain the basis of the maximum likelihood estimation.

\subsection{Maximum likelihood estimation for fundamental data model}

Let $\Theta$ denote the parameter space that is an open subset of $\mathbb{R}^{1\times n}$\footnote[1]{This assumption is made for ease of exposition.}. The maximum likelihood estimate $\hat{\theta}_{mle}$ of $\theta\in\Theta$ for the fundamental data model is given by
\begin{align*}
\hat{\theta}_{mle}=\argmin_{\theta\in\Theta}\Big(-\log\mathcal{L}_f(\theta|r_{1:L})\Big),
\end{align*}where $r_1,\cdots,r_L\in\mathbb{R}^2$ denote the acquired data and $\mathcal{L}_f(\theta|r_{1:L})=p^{\theta}_{R_{1:L}|N(t)}\left(r_{1:L}|L\right)$ denotes the likelihood function for the fundamental data model given by Eq. (\ref{likelihood1}).

In the rest of this paper, we only focus on the estimation of the parameters of the motion model, such as drift coefficient, diffusion coefficient and initial location of the molecule, i.e., we assume that $\Lambda$ and $f_x$ are independent of $\theta$. All parameters of the trajectory of the molecule have been encapsulated in the parameter vector $\theta$ and our approach does not have significant restrictions on which parameters can be included in this parameter vector. We also denote the dependence of the variable/function on the parameter vector $\theta$, by adding $\theta$ to its symbol as a superscript or subscript. For example, $p_{R_{1:L}|N(t)}$ and $F$ in Eq. (\ref{likelihood1}) are denoted by $p^{\theta}_{R_{1:L}|N(t)}$ and $F_{\theta}$, respectively.

\subsection{Maximum likelihood estimation for practical data model}

The maximum likelihood estimate $\hat{\theta}_{mle}$ of $\theta\in\Theta$ for the practical data model is given by
\begin{align*}
\hat{\theta}_{mle}=\argmin_{\theta\in\Theta}\Big(-\log\mathcal{L}_p(\theta|z_1,\cdots,z_K)\Big),
\end{align*}where $\left\{z_1,\cdots,z_K\right\}, z_1,\cdots,z_K=0,1,\cdots; L=\sum_{k=1}^Kz_k$, denotes an image with $K$ pixels and $\mathcal{L}_p$ denotes the likelihood function for the infinite practical data model given by, according to Theorem \ref{main},
\begin{small}
\begin{align}
\label{prac2}
&\mathcal{L}_p(\theta|z_1,\cdots,z_K)\nonumber\\
&\ \ \ \ =Pr^{\theta}\left[S_1=z_1,\cdots,S_K=z_K\right]\nonumber\\
&\ \ \ \ =p_L Pr^{\theta}\left[S_1=z_1,\cdots,S_K=z_K|N(t)=L\right]\nonumber\\
&\ \ \ \ =p_L\int_{\Delta^L}\Bigg[\sum_{v_{1:L}\in A^L_K\left(z_{1:K}\right)}\int_{\mathbb{R}^{3\times L}}\prod_{l=1}^LI_{C_{v_l}}(x_l)\nonumber\\
&\ \ \ \ \ \ \ \ \ \ \times p^{\theta}_{X(\tau_1),\cdots,X(\tau_L)}\left(x_{1:L}\right)dx_{1:L}\Bigg]p_{T_{1:L}|N(t)}\left(\tau_{1:L}|L\right)d\tau_{1:L}, 
\end{align}
\end{small}where
\begin{align}
I_{C_{v_l}}(x_l):=\int_{C_{v_l}}f_{x_l}(r)dr,\quad l=1,\cdots,L,
\end{align}is the probability that, given that the systems state is $x_l$ at time $\tau_l$, the photon that arrives as time $\tau_l$ hits $C_{v_l}$. For the practical data model, we have
\begin{footnotesize}
\begin{align}
\label{prac1}
&\mathcal{L}_p(\theta|z_1,\cdots,z_K)\nonumber\\
&\ \ \ \ =\sum_{z=0}^{\infty}p_{L+z}\int_{\Delta^{L+z}}\Bigg[\sum_{v_{1:L+z}\in A^{L+z}_{K+1}\left(z_{1:K},z\right)}\int_{\mathbb{R}^{3\times(L+z)}}\prod_{l=1}^{L+z}I_{C_{v_l}}(x_l)\nonumber\\
&\ \ \ \ \ \ \ \ \ \times p^{\theta}_{X(\tau_1),\cdots,X(\tau_{L+z})}\left(x_{1:L+z}\right)dx_{1:L+z}\Bigg]\nonumber\\
&\ \ \ \ \ \ \ \ \ \times p_{T_{1:L+z}|N(t)}\left(\tau_{1:L+z}|L+z\right)d\tau_{1:L+z}.
\end{align}
\end{footnotesize}

In general, computing the integrals of the likelihood function is not a trivial task. Here, based on the Monte Carlo approach provided in \cite{automatica,ref_monte,singh}, we develop an algorithm to approximate these integrals. The basis of our algorithm is the law of large numbers, which can be stated as follows.

Let $\left\{\tau^n_1,\cdots,\tau^n_L\right\}_{n=1}^N$ be $N$ sequences of $L$  time points in the interval $\left[t_0,t\right]$, each are arrivals of a Poisson process with intensity function $\Lambda(\psi)$ and such that precisely $L$ time points are in the exposure time interval $\left[t_0,t\right]$. Let $X(\tau^n_1),\cdots,X(\tau^n_L), n=1,\cdots,N$, be 3D random variables that describe the locations of the object at time points $\tau_0:=t_0\leq\tau^n_1<\cdots<\tau^n_L\leq t$. Let $p_{X(\tau^n_1),\cdots,X(\tau^n_L)}$ be the joint distribution of $X(\tau^n_1),\cdots,X(\tau^n_L)$. Then, the likelihood function (Eq. (\ref{prac2})) can be approximated as
\begin{small}
\begin{align}
\label{approximation_1}
&\mathcal{L}_p(\theta|z_1,\cdots,z_K)\nonumber\\
&\ \ \ \ \ \approx p_L\frac{1}{N}\sum_{n=1}^N\sum_{v_{1:L}\in A^L_K\left(z_{1:K}\right)}\left[\frac{1}{M_{c}}\sum_{m=1}^{M_{c}}\left(\prod_{l=1}^LI_{C_{v_l}}(x^{m,n}_l)\right)\right],
\end{align}
\end{small}where according to the law of large numbers
\begin{small}
\begin{align}
\label{approximation_2}
&\lim_{M_{c}\rightarrow\infty}\frac{1}{M_{c}}\sum_{m=1}^{M_{c}}\left(\prod_{l=1}^LI_{C_{v_l}}(x^{m,n}_l)\right)\nonumber\\
&\ \ \ \ \ =E\left\{\prod_{l=1}^LI_{C_{v_l}}(X(\tau_l))\right\}\nonumber\\
&\ \ \ \ \ =\int_{\mathbb{R}^{2\times L}}\left(\prod_{l=1}^LI_{C_{v_l}}(x_l)\right)p_{X(\tau^n_1),\cdots,X(\tau^n_L)}\left(x_{1:L}\right)dx_{1:L},
\end{align}
\end{small}in which $\left\{x^{m,n}:=\left(x^{m,n}_1,\cdots,x^{m,n}_L\right)\right\}_{m=1}^{M_{c}}, x^{m,n}_l\in\mathbb{R}^3, l=1,\cdots,L, m=1,\cdots,M_{c}$, is a sequence of ${M_{c}}$ independent and identically distributed trajectories drawn from the distribution $p_{X(\tau^n_1),\cdots,X(\tau^n_L)}$. As the state process $\left\{\tilde{X}(\tau)\right\}_{\tau\geq t_0}$ is a Markov process (see Eq. (\ref{sde1})) we carry out the Monte Carlo simulations for this state process and then obtain the simulations for $\left\{X(\tau)\right\}_{\tau\geq t_0}$ as $X(\tau)=H\tilde{X}(\tau), \tau\geq t_0$. As $\{\tilde{X}(\tau)\}_{\tau\geq t_0}$ is Markov, the joint probability density function $p_{\tilde{X}(\tau_1),\cdots,\tilde{X}(\tau_L)}$ of $\tilde{X}(\tau_1),\cdots,\tilde{X}(\tau_L)$ is given by
\begin{small}
\begin{align}
p_{\tilde{X}(\tau_1),\cdots,\tilde{X}(\tau_L)}\left(\tilde{x}_{1:L}\right)=p_{\tilde{X}(\tau_1)}\left(\tilde{x}_1\right)\prod_{l=2}^Lp_{\tilde{X}(\tau_l)|\tilde{X}(\tau_{l-1})}\left(\tilde{x}_l|\tilde{x}_{l-1}\right).
\end{align}
\end{small}
We draw $M_{c}$ trajectories $x^{m,n}, m=1,\cdots,M_{c}, n=1,\cdots,N$, through the following Monte Carlo algorithm \cite{automatica,ref_monte,singh}:

\begin{algorithm}[Monte Carlo method]
\label{algorithm1}
\textbf{Step 0.} Suppose we have observed the arrival of $L$ photons on the detector during the exposure time interval. Let $\left\{\tau^n_1,\cdots,\tau^n_L\right\}_{n=1}^N$ be $N$ sequences of $L$ time points, each are arrival times of a Poisson process with intensity function $\Lambda(\psi)$ and with precisely $L$ time points in the exposure time interval $\left[t_0,t\right]$.

\textbf{Step 1.} For each $n$, draw independent and identically distributed (i.i.d.) samples $\left\{\tilde{x}_1^{m,n}\right\}_{m=1}^{M_{c}}$ according to $p_{\tilde{X}(\tau^n_1)}(\tilde{x}), \tilde{x}\in\mathbb{R}^k$, i.e., $\tilde{x}_1^{m,n}\sim p_{\tilde{X}(\tau^n_1)}(\tilde{x}), m=1,\cdots,M_{c}$.

\textbf{Step 2.} For each $n$, draw i.i.d. samples $\left\{\tilde{x}_2^{m,n}\right\}_{m=1}^{M_{c}}$ according to $p_{\tilde{X}(\tau^n_2)|\tilde{X}(\tau^n_1)}\Big(\tilde{x}|\tilde{x}_1^{m,n}\Big), \tilde{x}\in\mathbb{R}^k$, i.e., $\tilde{x}_2^{m,n}\sim p_{\tilde{X}(\tau^n_2)|\tilde{X}(\tau^n_1)}\Big(\tilde{x}|\tilde{x}^{m,n}_1\Big), m=1,\cdots,M_c$.

\ \ \ \vdots

\textbf{Step $L$.} For each $n$, draw i.i.d. samples $\left\{\tilde{x}_L^{m,n}\right\}_{m=1}^{M_c}$ according to $p_{\tilde{X}(\tau^n_L)|\tilde{X}(\tau^n_{L-1})}\Big(\tilde{x}|\tilde{x}_{L-1}^{m,n}\Big), \tilde{x}\in\mathbb{R}^k$, i.e., $\tilde{x}_L^{m,n}\sim p_{\tilde{X}(\tau^n_L)|\tilde{X}(\tau^n_{L-1})}\Big(\tilde{x}|\tilde{x}^{m,n}_{L-1}\Big), m=1,\cdots,M_{c}$.

\textbf{Step $L+1$.} Repeat Steps $0$ to $L$, $N$ times and  approximate the likelihood function by Eqs. (\ref{approximation_1}) and (\ref{approximation_2}) by setting $x^{m,n}_l:=H\tilde{x}^{m,n}_l, m=1,\cdots,M_c, l=1,\cdots,L, n=1,\cdots,N$.
\end{algorithm} 
 
Note that with the same technique we can approximate  Eq. (\ref{prac1}). The infinite sum in Eq. (\ref{prac1}) can be approximated by truncating the summation. As the denominator of a term with index $z$ is $(L+z)!$ the summation can be expected to converge quickly.

We further assess the performance of the above algorithm in the computation of the likelihood function for the cases that we have a small number of photons detected in simulated pixelated images in Examples \ref{example1}, \ref{example2}, \ref{example3} and \ref{example4}. In Example \ref{example1}, we investigate the case that we have only one photon, and evaluate the convergence of the Monte Carlo approach with different numbers of samples. Example \ref{example2} is focused on the computation of maximum likelihood estimates of the unknown parameters of the motion model, and the evaluation of the mean and standard deviation of the estimates, for the cases that the number of photons detected in the simulated pixelated images is equal to four. In Example \ref{example3}, for data sets of simulated images of a stationary molecule, we show that the standard deviations of location estimates match well with the CRLB. Finally, in Example \ref{example4}, we examine the effect of pixel size of the detector on the performance of the estimation method, in terms of standard deviation, a fact that has not been considered in most of available methods. 

\begin{figure}[htbp]
\centering\includegraphics[width=0.5\textwidth]{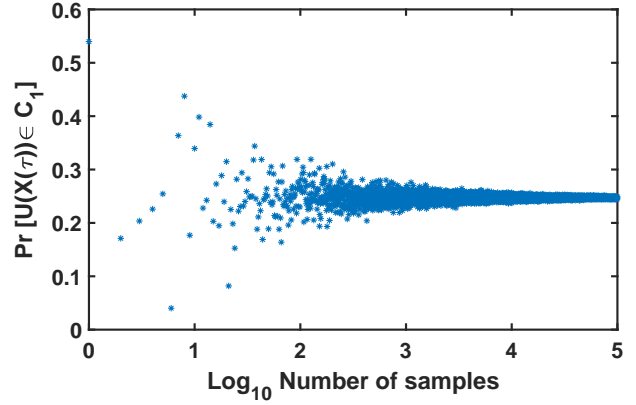}
\DeclareGraphicsExtensions{.eps} \caption{Convergence of the Monte Carlo method. The probabilities $Pr\left[U((X(\tau))\in C_1\right]=0.24735$ for different number $M_{c}$ of Monte Carlo samples, which ranges from 1 to 100,000, where $X(\tau)$ is a two-dimensional single molecule trajectory (Eq. (\ref{lineareq})), are shown in which the time point $\tau=0.01$ ms is fixed, with the first order drift coefficient $F = -10$/s and the diffusion coefficient $D = 1$ $\mu^2$/s ($G:=\sqrt{2D}$). Also, assume that the initial location of the molecule is $x(t_0) = (2.4, 2.4)^T$ $\mu$m. Detected locations of the photons emitted from the molecule in the image space are simulated using a zero-mean Gaussian model with covariance matrix $\Sigma=0.01I_{2\times 2}$ $\mu^2$m. A $60\times 60$ pixelated detector with square pixels of width $W=16$ $\mu$m is used to acquire the pixelated image of the molecule trajectory.}
\label{monte}
\end{figure}

\begin{example}
\label{example1}
Assume that we have a typical two-dimensional single molecule trajectory $X(\tau)$ in the object space (Eq. (\ref{lineareq})), where the time point $\tau=0.01$ ms is fixed, with the first order drift coefficient 
$F = -10$/s 
(we assume that there is no zero order drift, i.e., $V=0$) and the diffusion coefficient $D = 1$ $\mu^2$/s ($G:=\sqrt{2D}$). Also, we assume that the initial location of the molecule is $x(t_0) = (2.4, 2.4)^T$ $\mu$m. In the fundamental data model, detected locations of the photons emitted from the molecule in the image space are simulated using a zero-mean Gaussian profile with covariance matrix $\Sigma=0.01I_{2\times 2}$ $\mu^2$m. In the practical data model, a $60\times 60$ pixelated detector with square pixels of width $W=16$ $\mu$m is used to acquire the pixelated image of the molecule trajectory. Assume that the photon emitted from the object hits the pixel $C_1$ centered at $\left(c^1_x,c^1_y\right)=(230.75, 237.25)^T$ $\mu$m on the image plane. Then, using Algorithm \ref{algorithm1}, we calculate the probability that this event takes place as
\begin{small}
\begin{align*}
Pr\left[U((X(\tau))\in C_1\right]&=\int_{\mathbb{R}^2}I_{C_1}(x)p_{X(\tau)}(x)dx\nonumber\\
&\approx\frac{1}{M_{c}}\sum_{m=1}^{M_{c}}I_{C_1}\left(x^m\right),
\end{align*}
\end{small}where, for an invertible magnification matrix $M\in\mathbb{R}^{2\times 2}$ and $x\in\mathbb{R}^2$,
\begin{footnotesize}
\begin{align}
I_{C_1}(x)&=\frac{1}{\left|\det(M)\right|}\int_{C_1}q\left(M^{-1}r-x\right)dr\nonumber\\
&=\frac{1}{\left|\det(M)\right|}\int_{c^1_x-\frac{W}{2}}^{c^1_x+\frac{W}{2}}\int_{c^1_y-\frac{W}{2}}^{c^1_y+\frac{W}{2}}q\left(M^{-1}\left(r_x,r_y\right)-x\right)dr_ydr_x,
\end{align}
\end{footnotesize}and $\left\{x^m\right\}_{m=1}^{M_{c}}, x^m\in\mathbb{R}^2, m=1,\cdots,M_{c}$, is a sequence of independent and identically distributed samples drawn from the distribution $p_{X(\tau)}$ using Algorithm \ref{algorithm1}. In Fig. \ref{monte}, we have shown the probabilities $Pr\left[U((X(\tau))\in C_1\right]$ computed for different number $M_{c}$ of Monte Carlo samples. As can be seen in Fig. \ref{monte_2}, the standard deviation of the probabilities decreases by increasing the number of samples, where the number of samples ranges between (a) 1 to 25,000, (b) 25,001 to 50,000, (c) 50,001 to 75,000, and (d) 75,001 to 100,000 ((a) first, (b) second, (c) third, and (d) fourth quarters of the data. This suggests the convergence of these probabilities.
\end{example}

\begin{figure}[htbp]
\centering\includegraphics[width=0.5\textwidth]{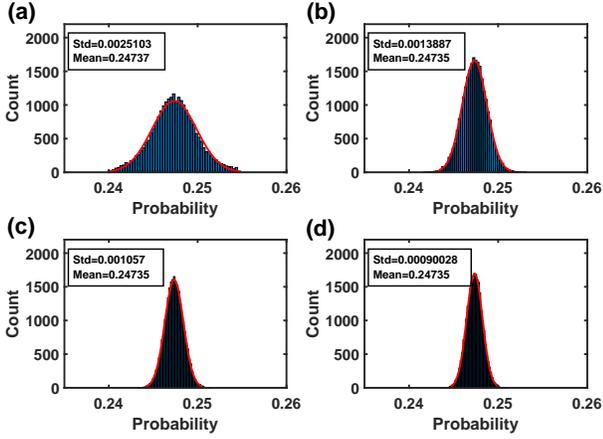}
\DeclareGraphicsExtensions{.eps} \caption{Histograms, means and standard deviations of the probabilities computed using the Monte Carlo method with different number of samples. Gaussian models fitted to the histograms of the probabilities computed using the Monte Carlo method, where the number of samples ranges between (a) 1 to 25,000, (b) 25,001 to 50,000, (c) 50,001 to 75,000, and (d) 75,001 to 100,000 ((a) first, (b) second, (c) third, and (d) fourth quarters of the data shown in Fig. \ref{monte}).}
\label{monte_2}
\end{figure}

\begin{example}
\label{example2}
We next examine the performance of our proposed parameter estimation method. For this purpose, we simulated 100 pixelated images of single molecule trajectories. These trajectories were simulated using Eq. (\ref{lineareq}) with four time points, where the time points were drawn from a Poisson process in the exposure time interval $[0,20]$ ms, and the first order drift coefficient $F=-10$/s and the diffusion coefficient $D=1.5$ $\mu\mbox{m}^2$/s ($G:=\sqrt{2D}$). Also, we assumed that the initial location of the molecule was $(2.3,2.3)^T$ $\mu$m. The locations of the photons emitted from the molecule trajectories, in the image space, were simulated with the Gaussian measurement noise (Eq. (\ref{gaussianimagefunction})) and $\sigma=0.1$ $\mu$m. We assumed that these photons were detected using a pixelated detector of pixel size and image size of $6.5 \times 6.5$ $\mu$m and $60\times 60$ pixels, respectively. For example, in the $50^{th}$ image, we have four photons which are detected in the pixels centered at $(217.75,230.75)^T, (172.25,217.75)^T, (146.25,204.75)^T$ and $(139.75,198.25)^T$ $\mu$m positions on the image plane, denoted by $C_1, C_2, C_3$ and $C_4$. Then, the summation of the likelihood function (Eq. \ref{approximation_1}) is performed over $A^4_4(1,1,1,1)=\left\{(1,2,3,4),(1,2,4,3),\cdots,(4,3,2,1)\right\}$ (the size of $A^4_4(1,1,1,1)$ in this case is equal to $4! = 24$).

We then estimated all parameters of the trajectories, e.g., initial location of the molecule, drift and diffusion coefficients, together using Algorithm \ref{algorithm1}, where the number of spatial $M_{c}$ and temporal $N$ Monte Carlo samples equal to 1000 and 100, respectively. The errors (estimate - true value) of the estimation are shown in Figs. \ref{locestimate} and \ref{driftdiffestimate}. As can be seen in these figures, the spreads of the errors are around zero and there is no systematic bias associated with the estimates.

\begin{figure}[htbp]
\centering\includegraphics[width=0.5\textwidth]{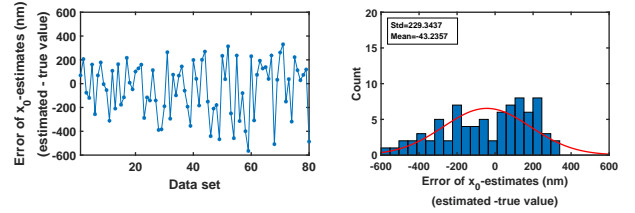}
\DeclareGraphicsExtensions{.eps} \caption[Analysis of the error of initial location estimates from pixelated images of single molecule trajectories for the Gaussian measurement noise case.] {Analysis of the error of initial location estimates from pixelated image data sets of single molecule trajectories for the Gaussian measurement noise case. Differences between the estimates of the initial $x_0$- and $y_0$-location of the molecule and their true values from the images of the molecule trajectories simulated using Eqs. (\ref{lineareq}) with four time points, where the time points are drawn from a Poisson process, and the first order drift coefficient $F=-10$/s and the diffusion coefficient $D=1.5$ $\mu\mbox{m}^2$/s ($G:=\sqrt{2D}$). The initial location of the molecule is $x(t_0)=(2.3,2.3)^T$ $\mu$m. The locations of the photons emitted from the molecule trajectories, in the image space, are simulated with the Gaussian measurement noise (Eq. (\ref{gaussianimagefunction})) and $\sigma=0.1$ $\mu$m. These photons are detected using a pixelated detector of pixel size and image size of $6.5 \times 6.5$ $\mu$m and $60\times 60$ pixels, respectively.}
\label{locestimate}
\end{figure}
\begin{figure}[htbp]
\centering\includegraphics[width=0.5\textwidth]{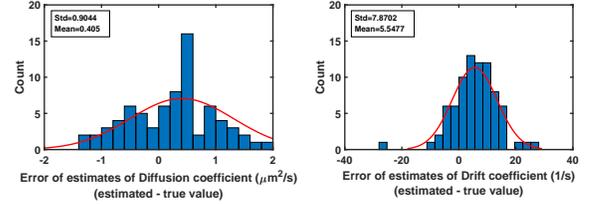}
\DeclareGraphicsExtensions{.eps} \caption[Analysis of the error of diffusion coefficient and drift coefficient estimates from pixelated images of single molecule trajectories for the Gaussian measurement noise case.] {Analysis of the error of diffusion coefficient and drift coefficient estimates from pixelated image data sets of single molecule trajectories for the Gaussian measurement noise case. Differences between the diffusion (first order drift) coefficient estimates and the true diffusion (first order drift) coefficient value for data sets of Fig. \ref{locestimate}.}
\label{driftdiffestimate}
\end{figure}

We also applied the algorithm to the pixelated images of single molecule trajectories simulated using an Airy point spread functions with $\alpha=\frac{2\pi n_a}{\lambda}=13.23$, which corresponds to a Gaussian profile with $\sigma=0.1$ $\mu$m. The parameters of the molecule trajectories were the same as the parameters of the data set of Fig. \ref{locestimate}. As can be seen in Figs. \ref{locestimate_airy} and \ref{driftdiffestimate_airy}, we have obtained similar results as in the Gaussian case.
\end{example}

\begin{figure}[htbp]
\centering\includegraphics[width=0.5\textwidth]{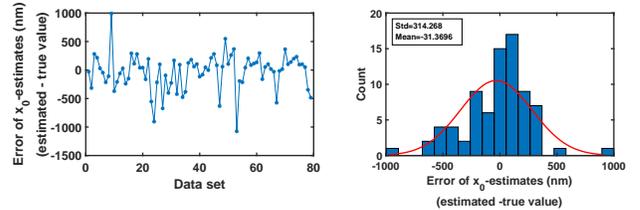}
\DeclareGraphicsExtensions{.eps} \caption[Analysis of the error of initial location estimates from pixelated images of single molecule trajectories for the Airy measurement noise case.] {Analysis of the error of initial location estimates from pixelated image data sets of single molecule trajectories for the Airy measurement noise case. Differences between the estimates of the initial $x_0$- and $y_0$-location of the molecule and their true values from the images of the molecule trajectories simulated using the parameters of the data set of Fig. \ref{locestimate}. The locations of the photons emitted from the molecule trajectories, in the image space, are simulated using an Airy model with $\alpha=\frac{2\pi n_a}{\lambda}=13.23$. These photons are detected using a pixelated detector of pixel size and image size of $6.5 \times 6.5$ $\mu$m and $60\times 60$ pixels, respectively.}
\label{locestimate_airy}
\end{figure}

\begin{figure}[htbp]
\centering\includegraphics[width=0.5\textwidth]{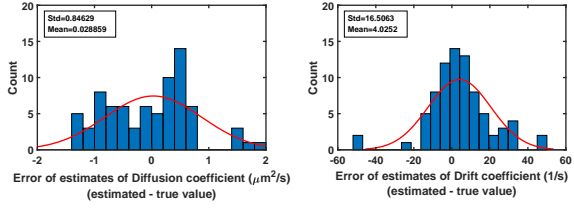}
\DeclareGraphicsExtensions{.eps} \caption[Analysis of the error of diffusion coefficient and drift coefficient estimates from pixelated images of single molecule trajectories for the Airy measurement noise case.] {Analysis of the error of diffusion coefficient and drift coefficient estimates from pixelated image data sets of single molecule trajectories for the Airy measurement noise case. Differences between the diffusion (first order drift) coefficient estimates and the true diffusion (first order drift) coefficient value for data sets of Fig. \ref{locestimate_airy}.}
\label{driftdiffestimate_airy}
\end{figure}

\begin{example}
\label{example3}
We further evaluate the performance of the proposed method in terms of the standard deviation of the estimates. In order to do this, we simulated the pixelated images of a stationary object using a pixelated detector of pixel size and image size of $6.5 \times 6.5$ $\mu$m and $60\times 60$ pixels, respectively, assuming that the number of photons detected by the detector is equal to three. The locations of the photons in the image space were simulated with the Gaussian measurement noise (Eq. (\ref{gaussianimagefunction})) and $\sigma=0.1$ $\mu$m. We then estimated the location of the molecule using Algorithm \ref{algorithm1}, where the number of Monte Carlo samples is equal to 10000. The errors of the location estimates are shown in Fig. \ref{loc_stationary}. As before, the errors are spreading around zero and no systematic bias can be seen. We also calculated the standard deviations of the estimates. These standard deviations, which are computed as 57.4 nm and 59.6 nm for the $x_0$- and $y_0$-locations of the molecule, respectively, are close to the localization accuracy, i.e., the positive definite square root of the CRLB, which is given as 58.37 nm for both $x$- and $y$-directions, reported in \cite{2004,anish,milad1,milad2}.
\end{example}

\begin{example}
\label{example4}
In this example, we investigate the effect of pixel size of the detector on the standard deviation of drift coefficient estimates. For this purpose, $100\times 100$ pixelated images were simulated using the parameters provided in Example \ref{example1}, assuming that four photons are detected in each image, and the pixel size ranges from $4 \times 4$ to $11.5 \times 11.5$ $\mu$m (for each pixel size, we have 50 images). As shown in the Figure \ref{pixel_size}, the standard deviation of the estimates gets worse as the pixel size of the detector increase.
\end{example}

\begin{figure}[htbp]
\centering\includegraphics[width=0.5\textwidth]{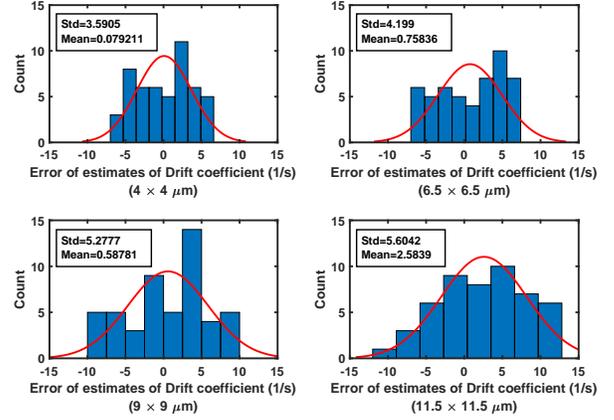}
\DeclareGraphicsExtensions{.eps} \caption[Analysis of the error of drift coefficient estimates from pixelated images with different pixel sizes.] {Analysis of the error of drift coefficient estimates from pixelated images with different pixel sizes. Analysis of the error of drift coefficient estimates from $100\times 100$ pixelated image data sets of single molecule trajectories for the Gaussian measurement noise case simulated using the parameters of Fig. \ref{locestimate}, assuming that four photons are detected in each image, where the pixel size ranges from $4 \times 4$ to $11.5 \times 11.5$ $\mu$m.}
\label{pixel_size}
\end{figure}

Here, we only consider a small number of photons, since, in general, the computation of the likelihood function (Eq. (\ref{practical_eq2})) is expensive. It is mostly because of the large number of members of the set $A^L_K\left(z_1,\cdots,z_K\right)$, which is equal to $\frac{L!}{z_1!\cdots z_K!}$, when $L$ increases. For example, in case of having a $32\times 32$-pixels detector with $L=1000$ and $K=1024, z_1=500,z_2=\cdots=z_{501}=1, z_{502}=\cdots=z_{1024}=0$, we have $\left|A^L_K\left(z_1,\cdots,z_K\right)\right|=\frac{1000!}{500!}=1000\times\cdots \times 501$, which is an extremely large number. To arrive at an estimator that can be practically computed, further research is needed for the cases in which the cardinality of the set $A_K^L\left(z_1,\cdots,z_K\right)$ is too large.
 
\begin{figure}[htbp]
\centering\includegraphics[width=0.5\textwidth]{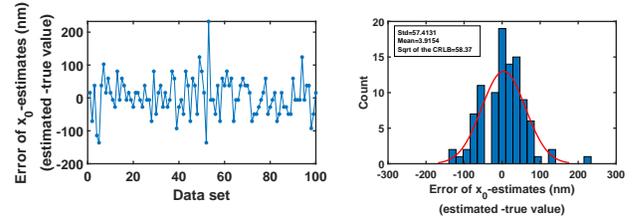}
\DeclareGraphicsExtensions{.eps} \caption[Analysis of the error of location estimates from pixelated images of a stationary molecule for the Gaussian measurement noise case.] {Analysis of the error of location estimates from pixelated image data sets of a stationary molecule for the Gaussian measurement noise case. Differences between the estimates of the initial $x_0$- and $y_0$-location of the molecule and their true values from the simulated images of a stationary molecule using a pixelated detector of pixel size and image size of $6.5 \times 6.5$ $\mu$m and $60\times 60$ pixels, respectively, assuming that three photons are detected by the detector. The locations of the photons in the image space are simulated with the Gaussian measurement noise (Eq. (\ref{gaussianimagefunction})) and $\sigma=0.1$ $\mu$m.}
\label{loc_stationary}
\end{figure}

\section{Fisher information matrix}
\label{crlbsection}

The Fisher information matrix was introduced in \cite{2004, 2006} for the analysis of microscopy image and data analysis problems, in particular for single molecule microscopy. It has since been used extensively for the evaluation of image analysis algorithms \cite{ed1, ed2}, experiment design \cite{mum2}, design of novel point spread functions \cite{Moerner_fish1, tahmasbi, Moerner_fish2} etc. The purpose of this section is to develop a framework for the computation of the Fisher information for the experimental setting considered here.

According to a well-known result in information theory, known as the Cram\'er-Rao inequality, the covariance matrix of any unbiased estimator $\hat{\theta}$ of an unknown parameter vector $\theta$ is bounded from below by the inverse of the Fisher information matrix $I(\theta)$, i.e., $\mbox{Cov}(\hat{\theta})\geq I^{-1}(\theta)$, assuming that $I(\theta)$ is invertible. \cite{jerry}. Then, the smallest standard deviation of the estimates that can be obtained, which is independent of the estimation method, only depends on the statistical model of the data, and is given by the positive definite square root of the inverse of the Fisher information matrix, referred to as the Cram\'er-Rao lower bound (CRLB). Large sample theory tells us that for a number of large repeats the maximum likelihood estimator has a variance that is given by the inverse of the Fisher information matrix \cite{jerry_fisher1,jerry_fisher2}. The usefulness of the CRLB also derives from the fact that in many practical circumstances for a maximum likelihood estimator it has been found that the variance of the estimator is well approximated by the inverse of the Fisher information matrix. In the following, for the fundamental and practical data models introduced in Sections \ref{fundamental} and \ref{practical}, we calculate the Fisher information matrix.

\subsection{Fisher information matrix for fundamental data model}

In this section, for the fundamental data model, we first, in Definition \ref{fish_fixed}, introduce the notation for the Fisher information matrix of the fundamental data model given the number of detected photons. 

\begin{definition}
\label{fish_fixed}
Let the parameter space $\Theta$ describe an open subset of $\mathbb{R}^{1\times n}$ containing the true parameters. For $L=1,2,\cdots$, and a row parameter vector $\theta\in\Theta$, we introduce the $n\times n$ Fisher information matrix of the fundamental data model given $N(t)=L$, as
\begin{scriptsize}
\begin{align}
\label{deffish}
I^f_{N(t)=L}(\theta):&=E_{p^{\theta}_{R_{1:L}|N(t)=L}}\Bigg\{\left(\frac{\partial \log p^{\theta}_{R_{1:L}|N(t)}\left(r_{1:L}|L\right)}{\partial\theta}\right)^T\nonumber\\
&\ \ \ \times\left(\frac{\partial \log p^{\theta}_{R_{1:L}|N(t)}\left(r_{1:L}|L\right)}{\partial\theta}\right)\Bigg\}\nonumber\\
&=\int_{\mathbb{R}^{2\times L}}p^{\theta}_{R_{1:L}|N(t)}\left(r_{1:L}|L\right)\left(\frac{\partial \log p^{\theta}_{R_{1:L}|N(t)}\left(r_{1:L}|L\right)}{\partial\theta}\right)^T\nonumber\\
&\ \ \ \times\left(\frac{\partial \log p^{\theta}_{R_{1:L}|N(t)}\left(r_{1:L}|L\right)}{\partial\theta}\right)dr_{1:L},
\end{align}
\end{scriptsize}where $E_{p^{\theta}_{R_{1:L}|N(t)=L}}$ is the expected value with respect to the probability density $p^{\theta}_{R_{1:L}|N(t)=L}$.
\end{definition}

In the above definition, we introduced the Fisher information matrix conditioned on the number of detected photons, which will be further used to derive an expression for the the Fisher information matrix of the practical data model. Note that the relationship between the actual (non-conditional) $I^f(\theta)$ and conditional $I^f_{N(t)=L}(\theta)$ Fisher information matrix for the fundamental data model is given by
\begin{footnotesize}
\begin{align}
I^f(\theta)&=\sum_{L=0}^{\infty}\int_{\mathbb{R}^{2\times L}}p^{\theta}_{R_{1:L}}\left(r_{1:L}\right)\left(\frac{\partial \log p^{\theta}_{R_{1:L}}\left(r_{1:L}\right)}{\partial\theta}\right)^T\nonumber\\
&\ \ \ \ \ \ \ \ \ \ \times\left(\frac{\partial \log p^{\theta}_{R_{1:L}}\left(r_{1:L}\right)}{\partial\theta}\right)dr_{1:L}\nonumber\\
&=\sum_{L=0}^{\infty}p_LI^f_{N(t)=L}(\theta).
\end{align}
\end{footnotesize}

In \cite{miladsiam}, we obtained a recursive formulation for the Fisher information matrix of the fundamental data model, which is computationally expensive for non-Gaussian measurements. In the following theorem, based on the new expression derived for the likelihood function in Theorem \ref{main0}, we provide a closed form expression for the Fisher information matrix.

\begin{theorem}
\label{main3}
For a row parameter vector $\theta\in\Theta$, the Fisher information matrix $I^f_{N(t)=L}(\theta), L=1,2,\cdots$, of the fundamental data model given $N(t)=L$, can be calculated as (we assume that $p^{\theta}_{R_{1:L}|N(t)}$ is strictly positive)

\begin{footnotesize}
\begin{align}
&I^f_{N(t)=L}(\theta)\nonumber\\
&\ \ \ \ \ =p_L\int_{\mathbb{R}^{2\times L}}\frac{DF^T_{\theta}\left(r_{1:L}\right)DF_{\theta}\left(r_{1:L}\right)}{\int_{\mathbb{R}^{3\times L}}\left(\int_{\Delta^L}F_{\theta}\left(r_{1:L},x_{1:L},\tau_{1:L}\right)d\tau_{1:L}\right)dx_{1:L}}dr_{1:L},
\end{align}
\end{footnotesize}where
\begin{small}
\begin{align}
\label{F}
DF_{\theta}\left(r_{1:L}\right):=\frac{\partial}{\partial\theta}\int_{\mathbb{R}^{3\times L}}\left(\int_{\Delta^L}F_{\theta}\left(r_{1:L},x_{1:L},\tau_{1:L}\right)d\tau_{1:L}\right)dx_{1:L},
\end{align}
\end{small}in which $F_{\theta}(.)$ is given by Eq. (\ref{F_theta}).
\end{theorem}

\begin{proof}
See Appendix \ref{proof_main3}. $\hfill\Box$ 
\end{proof}

The structure of the presentation of the Fisher information is interesting to observe. The Fisher information matrix $I^f_{N(t)=L}(\theta)$ is given as a multidimensional integral over the different photon impact locations on the detector of the integrand. Importantly, the same integrand will also play an important role in the expression of the Fisher information matrix for the practical data model, that we will consider next. 

%\textbf{As can be seen, Fisher information matrix depends on the derivatives of the probability of states.} In Lemma \ref{der_states} (see Appendix %\ref{der_states_section}), for an object's motion modeled by a linear stochastic system, we calculate the derivatives $Dp^{\theta}_{X(\tau_1),\cdots,X(\tau_L)}$ %in the Fisher information matrix (Eq. (\ref{fisher_new})) derived in the above theorem.\\

\subsection{Fisher information matrix for practical data model}

In this section, we use the results obtained in the previous section to calculate the Fisher information matrix for the practical data model. We first, in the following definition, introduce a notation for the Fisher information matrix of the practical data model.

\begin{definition}
Let the parameter space $\Theta$ describe an open subset of $\mathbb{R}^{1\times n}$ containing the true parameters. We introduce the following notation for the Fisher information matrix of the practical data model, for a row parameter vector $\theta\in\Theta$,
\begin{footnotesize}
\begin{align*}
I^p(\theta):&=E_{Pr^{\theta}\left[S_1=z_1,\cdots,S_K=z_K\right]}\nonumber\\
&\ \ \ \ \ \ \ \left\{\Bigg(\frac{\partial \log Pr^{\theta}\left[S_1=z_1,\cdots,S_K=z_K\right]}{\partial\theta}\right)^T\nonumber\\
&\ \ \ \ \ \ \ \times\left(\frac{\partial \log Pr^{\theta}\left[S_1=z_1,\cdots,S_K=z_K\right]}{\partial\theta}\right)\Bigg\}\nonumber\\
&=\sum_{z_1,\cdots,z_K=0}^{\infty}Pr^{\theta}\left[S_1=z_1,\cdots,S_K=z_K\right]\nonumber\\
&\ \ \ \ \ \ \ \times\left(\frac{\partial \log Pr^{\theta}\left[S_1=z_1,\cdots,S_K=z_K\right]}{\partial\theta}\right)^T\nonumber\\
&\ \ \ \ \ \ \ \times\left(\frac{\partial \log Pr^{\theta}\left[S_1=z_1,\cdots,S_K=z_K\right]}{\partial\theta}\right).
\end{align*}
\end{footnotesize}
\end{definition}

In the following theorem, we calculate the Fisher information matrix of the practical data model introduced in the above definition.

\begin{theorem}
\label{fishfish}
1. For a row parameter vector $\theta\in\Theta$, the Fisher information matrix $I^p(\theta)$ of the infinite practical data model can be calculated by Eq. (\ref{fish_prac_1}), where $DF_{\theta}$ and $F_{\theta}$ are given by Eqs. (\ref{F}) and (\ref{F_theta}).
\begin{figure*}
\begin{small}
\begin{align}
\label{fish_prac_1}
&I^p(\theta)=\sum_{L=0}^{\infty} p_L\sum_{z \in \mathbb{N}_0^K: |z|=L}\left(\frac{\sum_{v_{1:L}\in A^L_K\left(z\right)}\sum_{v'_{1:L}\in A^L_K\left(z\right)}\int_{C_{v_{1:L}}}\int_{C_{v'_{1:L}}}DF^T_{\theta}\left(r_{1:L}\right)DF_{\theta}\left(r'_{1:L}\right)dr'_{1:L}dr_{1:L}}{\sum_{v_{1:L}\in A^L_K\left(z\right)}\int_{C_{v_{1:L}}} \left[\int_{\mathbb{R}^{3\times L}}\Bigg(\int_{\Delta^L}F_{\theta}\left(r_{1:L},x_{1:L},\tau_{1:L}\right)d\tau_{1:L}\Bigg)dx_{1:L}\right]dr_{1:L}}\right).
\end{align}
\end{small}
\begin{footnotesize}
\begin{align}
\label{fish_prac_2}
&I^p(\theta)=\sum_{\bar{z}=0}^{\infty}\sum_{\bar{z}'=0}^{\infty}\sum_{L=0}^{\infty}p_{L+\bar{z}}p_{L+\bar{z}'}\sum_{z \in \mathbb{N}_0^K: |z|=L}\nonumber\\
&\left(\frac{\sum_{v_{1:L+\bar{z}}\in A^{L+\bar{z}}_{K+1}\left(z,\bar{z}\right)}\sum_{v'_{1:L+\bar{z}'}\in A^{L+\bar{z}'}_{K+1}\left(z,\bar{z}'\right)}\int_{C_{v_{1:L+\bar{z}}}}\int_{C_{v'_{1:L+\bar{z}'}}}  DF^T_{\theta}\left(r_{1:L+\bar{z}}\right)DF_{\theta}\left(r'_{1:L+\bar{z}'}\right)  dr'_{1:L+\bar{z}'}dr_{1:L+\bar{z}}}{\sum_{\bar{z}=0}^{\infty}p_{L+\bar{z}}\sum_{v_{1:L+\bar{z}}\in A^{L+\bar{z}}_{K+1}\left(z,\bar{z}\right)}\int_{C_{v_{1:L+\bar{z}}}} \left[\int_{\mathbb{R}^{3\times L}}\Bigg(\int_{\Delta^L}F_{\theta}\left(r_{1:L+\bar{z}},x_{1:L+\bar{z}},\tau_{1:L+\bar{z}}\right)d\tau_{1:L+\bar{z}}\Bigg)dx_{1:L+\bar{z}}\right]dr_{1:L+\bar{z}}}\right).
\end{align}
\end{footnotesize}
\begin{small}
\begin{align}
\label{fish_noise_1}
I^p(\theta)&=\int_{\mathbb{R}^L}\frac{\sum_{z_1,\cdots,z_K=0}^{\infty}\sum_{z'_1,\cdots,z'_K=0}^{\infty}p_{\mathcal{I}_{1:K}|S_{1:K}}\left(i_{1:K}|z_{1:K}\right)p_{\mathcal{I}_{1:K}|S_{1:K}}\left(i_{1:K}|z'_{1:K}\right)}{\sum_{z_1,\cdots,z_K=0}^{\infty}p_{\mathcal{I}_{1:K}|S_{1:K}}\left(i_{1:K}|z_{1:K}\right)Pr^{\theta}\left[S_1=z_1,\cdots,S_K=z_K\right]}\nonumber\\
&\ \ \ \ \ \ \ \ \ \ \ \ \ \ \ \ \ \ \ \ \ \times\left(\frac{\partial Pr^{\theta}\left[S_1=z_1,\cdots,S_K=z_K\right]}{\partial\theta}\right)^T\left(\frac{\partial Pr^{\theta}\left[S_1=z'_1,\cdots,S_K=z'_K\right]}{\partial\theta}\right)di_{1:K}.
\end{align}
\end{small}
\begin{footnotesize}
\begin{align}
\label{fish_noise_2}
&I^p(\theta):=\int_{\mathbb{R}^{L}}\sum_{L=0}^{\infty}\sum_{L'=0}^{\infty}p_Lp_{L'}\sum_{z \in \mathbb{N}_0^K: |z|=L}\sum_{z' \in \mathbb{N}_0^K: |z'|=L'}p_{\mathcal{I}_{1:K}|S_{1:K}}\left(i_{1:K}|z\right)p_{\mathcal{I}_{1:K}|S_{1:K}}\left(i_{1:K}|z'\right)\nonumber\\
&\times \left(\frac{\sum_{v_{1:L}\in A^L_K\left(z\right)}\sum_{v'_{1:L}\in A^{L'}_K\left(z'\right)}\int_{C_{v_{1:L}}}\int_{C_{v'_{1:L}}}DF^T_{\theta}\left(r_{1:L}\right)DF_{\theta}\left(r'_{1:L}\right)dr'_{1:L}dr_{1:L}}{\sum_{L=0}^{\infty}p_{L}\sum_{z \in \mathbb{N}_0^K: |z|=L}p_{\mathcal{I}_{1:K}|S_{1:K}}\left(i_{1:K}|z\right)           \sum_{v_{1:L}\in A^{L}_K\left(z\right)}\int_{C_{v_{1:L}}} \left[\int_{\mathbb{R}^{3\times L}}\Bigg(\int_{\Delta^{L}}F_{\theta}\left(r_{1:L},x_{1:L},\tau_{1:L}\right)d\tau_{1:L}\Bigg)dx_{1:L}\right]dr_{1:L}}\right)di_{1:K}.
\end{align}
\end{footnotesize}
\hrulefill
\end{figure*}

2. The Fisher information matrix $I^p(\theta)$ of the practical data model can be calculated by Eq. (\ref{fish_prac_2}).
\end{theorem}

\begin{proof}
See Appendix \ref{proof_practical_fim}. $\hfill\Box$ 
\end{proof}

As can be seen from the results of the above theorem, the key to computing the Fisher information expression is through the computation of the derivatives of the probability density function of the states. In \cite{hanzon_fisher}, for time-invariant systems, an easy-to-compute recursive formulation has been developed to deal with the derivatives of the probability density function of the states, and therefore, to compute the Fisher information matrix.

\section{Effect of noise}
\label{noise_effect_sec}

So far, we have assumed that all the photons detected by a pixelated detector come from the object of interest. However, in practice, fluorescence microscopy images always are corrupted by a background noise corresponding to the photons emitted from background components. The number of these photons in the $k^{th}, k=1,\cdots,K$, pixel is described by an independently Poisson distributed random variable $B_k$ with mean $\beta_k\geq 0$. Also, in a pixelated detector, the acquired image contains a readout noise, which can be modeled as an independently Gaussian distributed random variable $E_k$ with mean $\eta_k\geq 0$ and variance $\sigma_k^2>0$. The acquired image by a pixelated detector is then can be described by a
collection $\left\{\mathcal{I}^{\theta}_1,\cdots,\mathcal{I}^{\theta}_K\right\}$ of random variables given by
\begin{align*}
\mathcal{I}^{\theta}_k=S^{\theta}_k+B_k+E_k,\quad k=1,\cdots,K,\quad \theta\in\Theta.
\end{align*}
Note that $S_k^{\theta}$ and $B_k$ are non-negative integers, but $E_k$ is real-valued. Hence, $I_k^{\theta}$ is real-valued. In this case, the likelihood function $\mathcal{L}_p$ is given by
\begin{small}
\begin{align}
\label{prac3}
\mathcal{L}_p(\theta|i_1,\cdots,i_K)=p^{\theta}_{\mathcal{I}_1,\cdots,\mathcal{I}_K}\left(i_1,\cdots,i_K\right),\quad i_1,\cdots,i_K\in\mathbb{R},
\end{align}
\end{small}where $p^{\theta}_{\mathcal{I}_1,\cdots,\mathcal{I}_K}$ denotes the joint probability density function of $\mathcal{I}^{\theta}_1,\cdots,\mathcal{I}^{\theta}_K$, and can be calculated as
\begin{small}
\begin{align*}
p^{\theta}_{\mathcal{I}_{1:K}}\left(i_{1:K}\right)&=\sum_{z_1=0}^{\infty}\cdots\sum_{z_K=0}^{\infty}p_{\mathcal{I}_{1:K}|S_{1:K}}\left(i_{1:K}|z_{1:K}\right)\nonumber\\
&\ \ \ \ \ \ \ \ \ \ \times Pr^{\theta}\left[S_1=z_1,\cdots,S_K=z_K\right],
\end{align*}
\end{small}
in which the conditional probability density function of $\mathcal{I}_{1:K}$, given $S_{1:K}$, can be calculated as \cite{2006}
\begin{small}
\begin{align*}
&p_{\mathcal{I}_{1:K}|S_{1:K}}\left(i_{1:K}|z_{1:K}\right)\nonumber\\
&\ \ \ =\prod_{k=1}^Kp_{\mathcal{I}_k|S_k}\left(i_k|z_k\right)\nonumber\\
&\ \ \ =\prod_{k=1}^K\left[\frac{1}{\sqrt{2\pi}\sigma_k}\sum_{l=0}^{\infty}\left(\frac{(z_k+\beta_k)^le^{-(z_k+\beta_k)}}{l!}e^{-\frac{1}{2}\left(\frac{i_k-l-\eta_k}{\sigma_k}\right)^2}\right)\right],
\end{align*}
\end{small}
and $Pr^{\theta}\left[S_1=z_1,\cdots,S_K=z_K\right]$ is given by Eqs. (\ref{prac1}) or (\ref{prac2}). It has been shown that $I^p(\theta)$ can be calculated by Eq. (\ref{fish_noise_1}) (see Appendix \ref{fish_noise_2}). 

For the infinite practical data model, as calculated in Theorem \ref{fishfish}, the above expression can be rewritten as Eq. (\ref{fish_noise_2}).

The Fisher information matrix for an EMCCD detector can be obtained in a relatively straightforward fashion by applying the approaches developed for EMCCD detector in \cite{jerry_emccd} combined with the approach introduced here.

\section{Conclusions}
\label{conclusions}
The estimation of biophysical parameters from the observed trajectories of molecules in a live cell environment using single molecule fluorescence microscopy is one of the key experiments of modern molecular cell biology and biophysics. The approach we introduced here is one where we use a general stochastic dynamical system model to model the dynamics of the molecule, which includes dynamics governed by stochastic differential equations. The parameters of interest that are to be estimated can be any parameters that impact the underlying dynamical equations. Examples of such parameters are diffusion or drift coefficients, but also the coordinates of the particle at a particular point in time, such as the starting point of the trajectory. In a fluorescence microscopy experiment the acquired data is given by the photons that are emitted by the fluorescent label and are captured by the imaging detector. While the photons are assumed to be emitted based on a temporal Poisson process, a modern imaging detector has pixels and captures photons during an exposure interval. Such a detector therefore cannot capture the precise time points and impact locations of the detected photons and only records the accumulated photons in each pixel during the exposure time. Moreover, the impact points of the detected photons in the detector are related to the location of the imaged molecule through the so-called point spread function of the optical system that describes the image formation process of the microscope.

In this paper we have set up a stochastic framework within which the maximum likelihood estimator and Fisher information matrices could be derived for this parameter estimation problem. Central to our approach is a careful probabilistic model for the dynamics of the molecule and photon detection process for detection with a pixelated imaging detector. The resulting analytical expressions, derived without approximations from the general modeling assumptions, are complex  due to the intricate relationship between the statistics of the photons emission process by the imaged object and the photons statistics in the pixels of the detector.

Using a Monte Carlo approach we proposed a numerical method for computing the maximum likelihood  estimator. We showed, with simulated examples, that this estimator does indeed have desirable properties, such as low or no bias, even for extreme low photon count examples. The analysis of high photon count data poses significant numerical challenges as a large number of separate Monte Carlo simulations need to be carried out. To allow for efficient computations, this calls for further investigation into the numerics of the computation of the maximum likelihood estimator. Computation of the Fisher information is complicated through the need to compute a possible very large number of iterative integrals. As with the computation of the maximum likelihood estimator, it is Monte Carlo based methods that are expected to provide a successful approach. Using simulated data we illustrate that a detailed incorporation of the pixel size in the model for a parameter estimation problem is indeed of importance. In our example we show how the standard deviation of the estimate of the diffusion coefficient does depend on the pixel size.

We hope that the results that are presented here to provide important reference points for approximations that might lead to approaches that are computationally more efficient but need evaluation regarding their accuracy. Having a precise formulation available for the maximum likelihood estimator and the Fisher information matrices will hopefully help to provide important and well characterized tools to analyze molecular dynamics.

\section{Appendices}

\subsection{Analysis of Poisson time points}

\begin{lemma}
\label{Possion}
For $t_0\in\mathbb{R}$, let $\left\{N(\tau), \tau\geq t_0\right\}$ be a Poisson process with intensity function $\Lambda(\tau), \tau\geq t_0$. Let
$t_0\leq T_1<T2<\cdots$, be 1D random variables which describe ordered events of the process $N$. Then, the conditional probability density function $p_{T_1,\cdots,T_L|N(t)}$ of $T_1,\cdots,T_L$, given $N(t), t>t_0$, can be calculated as
\begin{align*}
p_{T_1,\cdots,T_L|N(t)}\left(\tau_{1:L}|L\right)=L!\frac{\prod_{l=1}^L\Lambda(\tau_l)}{\left(\int_{t_0}^t\Lambda(\psi)d\psi\right)^L}.
\end{align*}
\end{lemma}

\begin{proof}
See \cite{snyder}.$\hfill\Box$ 
\end{proof}

\subsection{Proof of Theorem \ref{main0}}
\label{proof_main0}
We have (see Lemma \ref{Possion} in Appendix for the probability density function of Poisson time points)
\begin{footnotesize}
\begin{align*}
&p_{R_{1:L}|N(t)}\left(r_{1:L}|L\right)\nonumber\\
&\ \ \ =\int_{\mathbb{R}^{3\times L}}p_{R_{1:L}|X(T_1),\cdots,X(T_L),N(t)}\left(r_{1:L}|x_{1:L},L\right)\nonumber\\
&\ \ \ \ \ \ \ \ \times p_{X(T_1),\cdots,X(T_L)|N(t)}\left(x_{1:L}|L\right)dx_{1:L}\nonumber\\
&\ \ \ =\int_{\mathbb{R}^{3\times L}}p_{R_1|X(T_1)}\left(r_1|x_1\right)\cdots p_{R_L|X(T_L)}\left(r_L|x_L\right)\nonumber\\
&\ \ \ \ \ \ \ \ \times\Bigg(\int_{\Delta^L}p_{X(T_1),\cdots,X(T_L),T_{1:L}|N(t)}\left(x_{1:L},\tau_{1:L}|L\right)d\tau_{1:L}\Bigg)dx_{1:L}\nonumber\\
&\ \ \ =\int_{\mathbb{R}^{3\times L}}p_{R_1}\left(r_1\right)\cdots p_{R_L}\left(r_L\right)\nonumber\\
&\ \ \ \ \ \ \ \ \times\Bigg(\int_{\Delta^L}p_{X(T_1),\cdots,X(T_L)|T_{1:L},N(t)}\left(x_{1:L}|\tau_{1:L},L\right)\nonumber\\
&\ \ \ \ \ \ \ \ \times p_{T_{1:L}|N(t)}\left(\tau_{1:L}|L\right)d\tau_{1:L}\Bigg)dx_{1:L}\nonumber\\
&\ \ \ =\int_{\mathbb{R}^{3\times L}}f_{x_1}\left(r_1\right)\cdots f_{x_L}\left(r_L\right)\nonumber\\
&\ \ \ \ \ \ \ \ \times\Bigg(\int_{\Delta^L}p_{X(\tau_1),\cdots,X(\tau_L)}\left(x_{1:L}\right)\frac{L!\prod_{l=1}^L\Lambda(\tau_l)}{\left(\int_{t_0}^t\Lambda(\psi)d\psi\right)^L}d\tau_{1:L}\Bigg)dx_{1:L}\nonumber\\
&\ \ \ =\int_{\mathbb{R}^{3\times L}}\Bigg[\int_{\Delta^L}\frac{L!}{\left(\int_{t_0}^t\Lambda(\psi)d\psi\right)^L}\nonumber\\
&\ \ \ \ \ \ \ \ \times\left(\prod_{i=1}^Lf_{x_i}(r_i)\Lambda(\tau_i)\right)p_{X(\tau_1),\cdots,X(\tau_L)}(x_{1:L})d\tau_{1:L}\Bigg]dx_{1:L}.
\end{align*}
\end{footnotesize}

\subsection{Proof of Theorem \ref{main}}
\label{proof_main}
1. According to the definitions of $S_1,\cdots,S_K$, we have, for $z_1,\cdots,z_K=0,1,\cdots$,
\begin{small}
\begin{align}
&Pr\left[S_1=z_1,\cdots,S_K=z_K\right]\nonumber\\
&\ \ \ =Pr\left[\bigcup_{v_{1:L}\in A^L_K\left(z_{1:K}\right)}\left\{\bigcap_{l=1}^L\left(U(X(T_l))\in C_{v_l}\right)\right\} | N(t)=L\right]\nonumber\\
&\ \ \ \ \ \ \ \ \times Pr\left[N(t)=L\right],
\end{align}
\end{small}where $L=\sum_{k=1}^Kz_{k}$. Since the events $\left\{\bigcap_{l=1}^L\left(U(X(T_l))\in C_{v_l}\right)\right\}$ are mutually exclusive, we have
\begin{small}
\begin{align*}
&Pr\left[\bigcup_{v_{1:L}\in A^L_K\left(z_{1:K}\right)}\left\{\bigcap_{l=1}^L\left(U(X(T_l))\in C_{v_l}\right)\right\} | N(t)=L\right]\nonumber\\
&\ \ \ =\sum_{v_{1:L}\in A^L_K\left(z_{1:K}\right)}Pr\left[\bigcap_{l=1}^L\left(U(X(T_l))\in C_{v_l}\right)|N(t)=L\right],
\end{align*}
\end{small}
and therefore, according to Eq. (\ref{likelihood1}),
\begin{small}
\begin{align}
&Pr\left[S_1=z_1,\cdots,S_K=z_K\right]\nonumber\\
&\ \ \ =\sum_{v_{1:L}\in A^L_K\left(z_{1:K}\right)}Pr\left[\bigcap_{l=1}^L\left(U(X(T_l))\in C_{v_l}\right)|N(t)=L\right]\nonumber\\
&\ \ \ \ \ \ \ \ \times Pr\left[N(t)=L\right]\nonumber\\
&\ \ \ =\sum_{v_{1:L}\in A^L_K\left(z_{1:K}\right)}p_L\int_{C_{v_{1:L}}}p_{R_{1:L}|N(t)}\left(r_{1:L}|L\right)dr_{1:L}\nonumber\\
&\ \ \ =p_L\sum_{v_{1:L}\in A^L_K\left(z_{1:K}\right)}\int_{C_{v_{1:L}}}\Bigg[\int_{\mathbb{R}^{3\times L}}\Bigg(\int_{\Delta^L}F\left(r_{1:L},x_{1:L},\tau_{1:L}\right)\nonumber\\
&\ \ \ \ \ \ \ \ \times d\tau_{1:L}\Bigg)dx_{1:L}\Bigg]dr_{1:L}.
\end{align}
\end{small}

2. We use the random variable $S_{K+1}$ to describe the number of photons in the complement pixel $C_{K+1}:=\mathbb{R}^2-\bigcup_{k=1}^K\bar{C_k}$ that result from the detection of the photons emitted from the object of interest. According to the definitions of $S_1,\cdots,S_{K+1}$, we have, for $z_1,\cdots,z_K=0,1,\cdots$,
\begin{small}
\begin{align}
&Pr\left[S_1=z_1,\cdots,S_K=z_K\right]\nonumber\\
&\ \ \ =\sum_{z=0}^{\infty}Pr\left[S_1=z_1,\cdots,S_K=z_K,S_{K+1}=z\right]\nonumber\\
&\ \ \ =\sum_{z=0}^{\infty}Pr\left[\bigcup_{v_{1:L+z}\in A^{L+z}_{K+1}\left(z_{1:K},z\right)}\left\{\bigcap_{l=1}^{L+z}\left(U(X(T_l))\in C_{v_l}\right)\right\}\right],
\end{align}
\end{small}where the inner summation term can be calculated from the similar approach used in part 1.

\subsection{Proof of Theorem \ref{main3}}
\label{proof_main3}
For a row parameter vector $\theta\in\Theta$, the Fisher information matrix $I^f_{N(t)=L}(\theta), L=1,2,\cdots$, given $N(t)=L$, can be calculated as, according to Eq. (\ref{deffish}) of Definition \ref{fish_fixed},
\begin{small}
\begin{align}
\label{an}
&I^f_{N(t)=L}(\theta)\nonumber\\
&\ \ \ =\int_{\mathbb{R}^{2\times L}}\frac{1}{p^{\theta}_{R_{1:L}|N(t)}\left(r_{1:L}|L\right)}\left(\frac{\partial p^{\theta}_{R_{1:L}|N(t)}\left(r_{1:L}|L\right)}{\partial\theta}\right)^T\nonumber\\
&\ \ \ \ \ \ \ \ \times\left(\frac{\partial p^{\theta}_{R_{1:L}|N(t)}\left(r_{1:L}|L\right)}{\partial\theta}\right)dr_{1:L}.
\end{align}
\end{small}
By substituting Eq. (\ref{likelihood1}) into Eq. (\ref{an}), we have
\begin{small}
\begin{align}
&I^f_{N(t)=L}(\theta)\nonumber\\
&\ \ =p_L\int_{\mathbb{R}^{2\times L}}\frac{1}{\int_{\mathbb{R}^{3\times L}}\left(\int_{\Delta^L}F_{\theta}\left(r_{1:L},x_{1:L},\tau_{1:L}\right)d\tau_{1:L}\right)dx_{1:L}}\nonumber\\
&\ \ \ \ \times\left(\frac{\partial}{\partial\theta}\int_{\mathbb{R}^{3\times L}}\left(\int_{\Delta^L}F_{\theta}\left(r_{1:L},x_{1:L},\tau_{1:L}\right)d\tau_{1:L}\right)dx_{1:L}\right)^T\nonumber\\
&\ \ \ \ \times\left(\frac{\partial}{\partial\theta}\int_{\mathbb{R}^{3\times L}}\left(\int_{\Delta^L}F_{\theta}\left(r_{1:L},x_{1:L},\tau_{1:L}\right)d\tau_{1:L}\right)dx_{1:L}\right)dr_{1:L}\nonumber\\
&\ \ =p_L\int_{\mathbb{R}^{2\times L}}\frac{DF^T_{\theta}\left(r_{1:L}\right)DF_{\theta}\left(r_{1:L}\right)}{\int_{\mathbb{R}^{3\times L}}\left(\int_{\Delta^L}F_{\theta}\left(r_{1:L},x_{1:L},\tau_{1:L}\right)d\tau_{1:L}\right)dx_{1:L}}dr_{1:L},
\end{align}
\end{small}where $DF_{\theta}$ is given by Eq. (\ref{F}).

\subsection{Proof of Theorem \ref{fishfish}}
\label{proof_practical_fim}

1. For a row parameter vector $\theta\in\Theta$, the Fisher information matrix $I^p(\theta)$ of the infinite practical data model can be calculated by Eq. (\ref{prac_fim_eq}). By substituting Eq. (\ref{practical_eq2}) into Eq. (\ref{prac_fim_eq}), we have Eq. (\ref{prac_fim_eq_2}), where $DF_{\theta}$ and $F_{\theta}$ are given by Eqs. (\ref{F}) and (\ref{F_theta}).

\begin{figure*}
\begin{footnotesize}
\begin{align}
\label{prac_fim_eq}
I^p(\theta)=\sum_{z_1=0}^{\infty}\cdots\sum_{z_K=0}^{\infty}\frac{1}{Pr^{\theta}\left[S_1=z_1,\cdots,S_K=z_K\right]}\left(\frac{\partial Pr^{\theta}\left[S_1=z_1,\cdots,S_K=z_K\right]}{\partial\theta}\right)^T\left(\frac{\partial Pr^{\theta}\left[S_1=z_1,\cdots,S_K=z_K\right]}{\partial\theta}\right).
\end{align}
\end{footnotesize}
\begin{footnotesize}
\begin{align}
\label{prac_fim_eq_2}
I^p(\theta)&=p_L\sum_{z_1=0}^{\infty}\cdots\sum_{z_K=0}^{\infty}\Bigg(\sum_{v_{1:L}\in A^L_K\left(z_{1:K}\right)}\sum_{v'_{1:L}\in A^L_K\left(z_{1:K}\right)}\int_{C_{v_{1:L}}}\int_{C_{v'_{1:L}}}\Bigg\{\int_{\mathbb{R}^{3\times L}}\int_{\mathbb{R}^{3\times L}}\Bigg[\int_{\Delta^L}\int_{\Delta^L}\Bigg.\Bigg.\nonumber\\
&\ \ \ \ \ \ \ \ \ \ \ \ \ \ \ \ \ \ \ \Bigg.\Bigg.\frac{\left(\frac{\partial}{\partial\theta}F_{\theta}\left(r_{1:L},x_{1:L},\tau_{1:L}\right)\right)^T \left(\frac{\partial}{\partial\theta}F_{\theta}\left(r'_{1:L},x'_{1:L},\tau'_{1:L}\right)\right)}{\sum_{v_{1:L}\in A^L_K\left(z\right)}\int_{C_{v_{1:L}}} \left[\int_{\mathbb{R}^{3\times L}}\Bigg(\int_{\Delta^L}F_{\theta}\left(r_{1:L},x_{1:L},\tau_{1:L}\right)d\tau_{1:L}\Bigg)dx_{1:L}\right]dr_{1:L}}d\tau_{1:L}d\tau'_{1:L}\Bigg]dx_{1:L}dx'_{1:L}\Bigg\}dr'_{1:L}dr_{1:L}\Bigg)\nonumber\\
&=\sum_{L=0}^{\infty} p_L\sum_{z \in \mathbb{N}_0^K: |z|=L}\left(\frac{\sum_{v_{1:L}\in A^L_K\left(z\right)}\sum_{v'_{1:L}\in A^L_K\left(z\right)}\int_{C_{v_{1:L}}}\int_{C_{v'_{1:L}}}DF^T_{\theta}\left(r_{1:L}\right)DF_{\theta}\left(r'_{1:L}\right)dr'_{1:L}dr_{1:L}}{\sum_{v_{1:L}\in A^L_K\left(z\right)}\int_{C_{v_{1:L}}} \left[\int_{\mathbb{R}^{3\times L}}\Bigg(\int_{\Delta^L}F_{\theta}\left(r_{1:L},x_{1:L},\tau_{1:L}\right)d\tau_{1:L}\Bigg)dx_{1:L}\right]dr_{1:L}}\right).
\end{align}
\end{footnotesize}
\begin{footnotesize}
\begin{align}
\label{noisy_fim_eq}
I^p(\theta):&=E_{Pr^{\theta}\left[\mathcal{I}_1=i_1,\cdots,\mathcal{I}_K=i_K\right]}\left\{\left(\frac{\partial \log Pr^{\theta}\left[\mathcal{I}_1=i_1,\cdots,\mathcal{I}_K=i_K\right]}{\partial\theta}\right)^T\left(\frac{\partial \log Pr^{\theta}\left[\mathcal{I}_1=i_1,\cdots,\mathcal{I}_K=i_K\right]}{\partial\theta}\right)\right\}\nonumber\\
&=\int_{\mathbb{R}^K}Pr^{\theta}\left[\mathcal{I}_1=i_1,\cdots,\mathcal{I}_K=i_K\right]\left(\frac{\partial \log Pr^{\theta}\left[\mathcal{I}_1=i_1,\cdots,\mathcal{I}_K=i_K\right]}{\partial\theta}\right)^T\left(\frac{\partial \log Pr^{\theta}\left[\mathcal{I}_1=i_1,\cdots,\mathcal{I}_K=i_K\right]}{\partial\theta}\right)di_{1:K}\nonumber\\
&=\int_{\mathbb{R}^K}\frac{1}{Pr^{\theta}\left[\mathcal{I}_1=i_1,\cdots,\mathcal{I}_K=i_K\right]}\left(\frac{\partial Pr^{\theta}\left[\mathcal{I}_1=i_1,\cdots,\mathcal{I}_K=i_K\right]}{\partial\theta}\right)^T\left(\frac{\partial Pr^{\theta}\left[\mathcal{I}_1=i_1,\cdots,\mathcal{I}_K=i_K\right]}{\partial\theta}\right)di_{1:K}\nonumber\\
&=\int_{\mathbb{R}^K}\frac{1}{Pr^{\theta}\left[\mathcal{I}_1=i_1,\cdots,\mathcal{I}_K=i_K\right]}\left(\sum_{z_1,\cdots,z_K=0}^{\infty}p_{\mathcal{I}_{1:K}|S_{1:K}}\left(i_{1:K}|z_{1:K}\right)\frac{\partial Pr^{\theta}\left[S_1=z_1,\cdots,S_K=z_K\right]}{\partial\theta}\right)^T\nonumber\\
&\ \ \ \ \ \ \ \ \ \ \ \ \ \ \ \ \ \ \ \ \ \times\left(\sum_{z'_1,\cdots,z'_K=0}^{\infty}p_{\mathcal{I}_{1:K}|S_{1:K}}\left(i_{1:K}|z'_{1:K}\right)\frac{\partial Pr^{\theta}\left[S_1=z'_1,\cdots,S_K=z'_K\right]}{\partial\theta}\right)di_{1:K}\nonumber\\%%%%
&=\int_{\mathbb{R}^K}\frac{\sum_{z_1,\cdots,z_K=0}^{\infty}\sum_{z'_1,\cdots,z'_K=0}^{\infty}p_{\mathcal{I}_{1:K}|S_{1:K}}\left(i_{1:K}|z_{1:K}\right)p_{\mathcal{I}_{1:K}|S_{1:K}}\left(i_{1:K}|z'_{1:K}\right)}{\sum_{z_1,\cdots,z_K=0}^{\infty}p_{\mathcal{I}_{1:K}|S_{1:K}}\left(i_{1:K}|z_{1:K}\right)Pr^{\theta}\left[S_1=z_1,\cdots,S_K=z_K\right]}\nonumber\\
&\ \ \ \ \ \ \ \ \ \ \ \ \ \ \ \ \ \ \ \ \ \times\left(\frac{\partial Pr^{\theta}\left[S_1=z_1,\cdots,S_K=z_K\right]}{\partial\theta}\right)^T\left(\frac{\partial Pr^{\theta}\left[S_1=z'_1,\cdots,S_K=z'_K\right]}{\partial\theta}\right)di_{1:K}.
\end{align}
\end{footnotesize}
\hrulefill
\end{figure*}

2. It results from the similar approach used in part 1.

\subsection{Effect of noise on Fisher information matrix}
\label{noisy_fim}

In the following, we show the derivation process of Eq. (\ref{noisy_fim_eq}). By substituting Eq. (\ref{prac3}) into the general equation of the Fisher information matrix, we then can obtain the Fisher information expression in terms of $Pr^{\theta}\left[S_1=z_1,\cdots,S_K=z_K\right]$ by the process shown by Eq. (\ref{noisy_fim_eq}).

\bibliographystyle{ieeetr}
\bibliography{MiladVahid}

\begin{comment}

\end{comment}

% biography section
% 
% If you have an EPS/PDF photo (graphicx package needed) extra braces are
% needed around the contents of the optional argument to biography to prevent
% the LaTeX parser from getting confused when it sees the complicated
% \includegraphics command within an optional argument. (You could create
% your own custom macro containing the \includegraphics command to make things
% simpler here.)
%\begin{IEEEbiography}[{\includegraphics[width=1in,height=1.25in,clip,keepaspectratio]{mshell}}]{Michael Shell}
% or if you just want to reserve a space for a photo:

%\begin{IEEEbiography}{Milad R. Vahid}
%Biography text here.
%\end{IEEEbiography}

% if you will not have a photo at all:
%\begin{IEEEbiographynophoto}{Bernard Hanzon}
%Biography text here.
%\end{IEEEbiographynophoto}

% insert where needed to balance the two columns on the last page with
% biographies
%\newpage

%\begin{IEEEbiographynophoto}{Raimund J. Ober}
%Biography text here.
%\end{IEEEbiographynophoto}

% You can push biographies down or up by placing
% a \vfill before or after them. The appropriate
% use of \vfill depends on what kind of text is
% on the last page and whether or not the columns
% are being equalized.

%\vfill

% Can be used to pull up biographies so that the bottom of the last one
% is flush with the other column.
%\enlargethispage{-5in}

% that's all folks
\end{document}